\documentclass[10pt,journal,compsoc]{IEEEtran}
\usepackage{amsmath,amssymb,amsfonts}
\pdfoutput=1
\usepackage{graphicx}
\usepackage{textcomp}
\usepackage{pifont} 
\usepackage{amssymb} 
\usepackage{xcolor}
\usepackage{array}
\usepackage[noend,ruled, linesnumbered, vlined]{algorithm2e}
\usepackage{textcomp}
\usepackage{multirow}
\usepackage{bm}
\usepackage{booktabs}
\usepackage{ntheorem}
\usepackage{subfigure} 
\usepackage{epstopdf}
\usepackage{cases}
\newenvironment{sequation}{\begin{equation}\small}{\end{equation}}

\usepackage{makecell,multirow,diagbox}
\usepackage{stfloats}
\usepackage{comment}
\usepackage{xcolor}
\usepackage{algpseudocode}
\usepackage{svg} 
\newenvironment{proof}{{\noindent\it Proof. }}{\hfill $\blacksquare$\par}
\usepackage{enumitem}[topsep=3pt,itemsep=3pt]
\allowdisplaybreaks[4]

\def\BibTeX{{\rm B\kern-.05em{\sc i\kern-.025em b}\kern-.08em
		T\kern-.1667em\lower.7ex\hbox{E}\kern-.125emX}}

\usepackage{xr}
\makeatletter
\newcommand*{\addFileDependency}[1]{
  \typeout{(#1)}
  \@addtofilelist{#1}
  \IfFileExists{#1}{}{\typeout{No file #1.}}
}
\makeatother
\newcommand*{\myexternaldocument}[1]{
    \externaldocument{#1}
    \addFileDependency{#1.tex}
    \addFileDependency{#1.aux}
}
\myexternaldocument{Appendices}

\definecolor{b}{rgb}{0.0, 0, 0}
\definecolor{b1}{rgb}{0.0, 0, 0}
\definecolor{b2}{rgb}{0.0, 0, 0}

\ifCLASSOPTIONcompsoc
\usepackage[nocompress]{cite}
\else
\usepackage{cite}
\fi

\theoremseparator{.}
\newskip\theorempreskipamount
\newskip\theorempostskipamount
\newtheorem{theorem}{Theorem}
\newtheorem{definition}{Definition}
\newtheorem{lemma}{Lemma}
\theorembodyfont{}
\newtheorem{remark}{Remark}

\begin{document}

\title{Joint Computing Resource Allocation and Task Offloading in Vehicular Fog Computing Systems Under Asymmetric Information}

\author{Geng~Sun,~\IEEEmembership{Senior~Member,~IEEE},
     Siyi~Chen,
     Zemin Sun,~\IEEEmembership{Member,~IEEE}, 
     Long He,
     Jiacheng Wang, \\ 
     Dusit Niyato,~\IEEEmembership{Fellow,~IEEE}, 
     Zhu Han,~\IEEEmembership{Fellow,~IEEE}, and 
     Dong In Kim,~\IEEEmembership{Life Fellow,~IEEE}

    \thanks{This study is supported in part by the National Natural Science Foundation of China (62272194, 62471200, 6250076101), in part by the Science and Technology Development Plan Project of Jilin Province (20250101027JJ), \textcolor{b2}{in part by the National Science Foundation (NSF) (ECCS-2302469), in part by Amazon, in part by Japan Science and Technology Agency (JST) Adopting Sustainable Partnerships for Innovative Research Ecosystem (ASPIRE) (JPMJAP2326), in part by Seatrium New Energy Laboratory, in part by the Singapore Ministry of Education (MOE) Tier 1 (RT5/23, RG24/24), in part by the Nanyang Technological University (NTU) Centre for Computational Technologies in Finance (NTU-CCTF), and in part by the Research Innovation and Enterprise (RIE) 2025 Industry Alignment Fund - Industry Collaboration Projects (IAF-ICP) (Award I2301E0026), administered by A*STAR.} (\textit{Corresponding author: Zemin Sun.})}
     \IEEEcompsocitemizethanks
     {
     \IEEEcompsocthanksitem Geng Sun is with the College of Computer Science and Technology, Key Laboratory of Symbolic Computation and Knowledge Engineering of Ministry of Education, Jilin University, Changchun 130012, China, and also with the College of Computing and Data Science, Nanyang Technological University, Singapore 639798 (e-mail: sungeng@jlu.edu.cn).
     \IEEEcompsocthanksitem Siyi Chen, Zemin Sun, and Long He are with the College of Computer Science and Technology, Key Laboratory of Symbolic Computation and Knowledge Engineering of Ministry of Education, Jilin University, Changchun 130012, China (e-mails: sychen23@mails.jlu.edu.cn, sunzemin@jlu.edu.cn, and helong0517@foxmail.com).
     \IEEEcompsocthanksitem Jiacheng Wang and Dusit Niyato are with the College of Computing and Data Science, Nanyang Technological University, Singapore 639798 (e-mails: jiacheng.wang@ntu.edu.sg and dniyato@ntu.edu.sg).
    \IEEEcompsocthanksitem \textcolor{b2}{Zhu Han is with the Department of Electrical and Computer Engineering at the University of Houston, Houston, TX 77004 USA, and also with the Department of Computer Science and Engineering, Kyung Hee University, Seoul, South Korea, 446-701} (e-mail: hanzhu22@gmail.com).
     \IEEEcompsocthanksitem Dong In Kim is with the Department of Electrical and Computer Engineering, Sungkyunkwan University, Suwon 16419, South Korea (email: dongin@skku.edu).
     }
}

\IEEEtitleabstractindextext{%
\begin{abstract}	
\par Vehicular fog computing (VFC) has emerged as a promising paradigm, which leverages the idle computational resources of nearby fog vehicles (FVs) to complement the computing capabilities of conventional vehicular edge computing. However, utilizing VFC to meet the delay-sensitive and computation-intensive requirements of the FVs poses several challenges. First, the limited resources of road side units (RSUs) struggle to accommodate the growing and diverse demands of vehicles. This limitation is further exacerbated by the information asymmetry between the controller and FVs due to the reluctance of FVs to disclose private information and to share resources voluntarily. This information asymmetry hinders the efficient resource allocation and coordination. Second, the heterogeneity in task requirements and the varying capabilities of RSUs and FVs complicate efficient task offloading, thereby resulting in inefficient resource utilization and potential performance degradation. To address these challenges, we first present a hierarchical VFC architecture that incorporates the computing capabilities of both RSUs and FVs. Then, we formulate a delay minimization optimization problem (DMOP), which is an NP-hard mixed integer nonlinear programming (MINLP) problem. To solve the DMOP, we propose a joint computing resource allocation and task offloading approach (JCRATOA), which comprises the components of computing resource allocation and task offloading. Specifically, we propose a convex optimization-based method for RSU resource allocation and a contract theory-based incentive mechanism for FV resource allocation. Moreover, we present a two-sided matching method for task offloading by employing the matching game. Additionally, we theoretically prove the polynomial complexity of JCRATOA. Simulation results demonstrate that the proposed JCRATOA outperforms the benchmark approaches, \color{b1}{achieving at least 7.6\%, 6.6\%, 6.25\%, and 11.9\% improvements in terms of the task completion delay, task completion ratio, system throughput, and resource utilization fairness, respectively, while satisfying the energy constraints of task vehicles (TVs), RSUs, and FVs.}
\end{abstract}
	
}

\maketitle
\IEEEdisplaynontitleabstractindextext
\IEEEpeerreviewmaketitle

\section{Introduction}
\label{sec_introduction}

\par The proliferation of road vehicles and the advancement of vehicular networks are propelling the emergence of various vehicular applications such as real-time navigation and collision detection. Moreover, the development of generative AI~\cite{Sun2025GenAI} further accelerates these advancements, thus enabling functionalities such as predictive maintenance and natural language-based emergency response~\cite{Sun2024LLM}. Most of these applications are computation-hungry and delay-sensitive, which drive unprecedented requirements for computing resources to satisfy the ultralow task execution delay. However, conventional cloud computing struggles to meet the stringent requirements of these applications due to the long communication distance between vehicles and remote cloud servers. To address this challenge, multi-access computing (MEC) has been regarded as a promising technology by migrating the cloud computing capabilities to road side units (RSUs) in close proximity to vehicles, thereby driving the advancement vehicular edge computing. By offloading the tasks to adjacent RSUs equipped with MEC servers, the computation performance of vehicles can be significantly extended in a low-latency and cost-effective way.

\par The geographic-random and time-varying requirements of vehicles make it challenging to deploy a sufficient number of MEC servers at a low cost. On the one hand, an inadequate number of MEC servers can lead to server overloads at certain MEC servers, particularly during peak periods. On the other hand, densely deploying MEC servers can incur high installation costs and result in the resource wastage during off-peak times. To address this limitation, vehicular fog computing (VFC) offers a promising solution by leveraging the underutilized computing resources of nearby vehicles~\cite{Wu2024Delay}. Specifically, the vehicles with idle resources serve as fog vehicles (FVs) to assist the RSUs in task processing, especially as future intelligent vehicles are expected to be equipped with more powerful onboard computing units~\cite{Liu2023Mobility}. Therefore, the tasks of a vehicle can be offloaded to a neighboring FV when the direct connectivity to RSUs is not possible, or when the RSU in range is overloaded. Despite the above mentioned advantages, the widespread deployment of VFC faces several challenges.

\textbf{Resource allocation.} First, compared to the cloud with rich resources, the computing resources of RSUs are limited. Without efficient resource allocation, it is difficult for an RSU to meet the computation-hungry and latency-sensitive demands of multiple vehicles simultaneously, especially during the peak hours\textcolor{b}{~\cite{Liu2025JointRobust}}. Moreover, due to the costs of task processing and the risks of privacy leakage, self-interested and privacy-aware vehicles are often reluctant to reveal the private information, such as their resource availability and willingness to collaborate. This reluctance causes information asymmetry between the controller and FVs, as the MBS lacks accurate and real-time knowledge of the resource status or intentions of FVs\textcolor{b}{~\cite{Lu2025Enhancing}}. Consequently, this information asymmetry further discourages FVs from voluntarily sharing resources or collaborating efficiently, ultimately leading to the inefficient utilization of the available resources\textcolor{b}{~\cite{Chen2023Comprehensive}}. Consequently, efficiently utilizing the computing resources of RSUs and FVs to meet the stringent demands of vehicles remains a significant challenge. 

\textbf{Task offloading.} Different vehicles generate tasks with diverse computational demands, while different RSUs and FVs possess varying processing capabilities\textcolor{b}{~\cite{Chen2025Efficient}}. Without an optimized offloading method, the RSUs and FVs are often under-loaded or over-loaded, which leads to poor task processing performance and inadequate resource utilization\textcolor{b}{~\cite{Hevesli2024Task}}. For example, a resource-hungry task may be offloaded to an FV that lacks sufficient computational power, thus causing delays or task failure, while more capable RSUs or FVs remain underutilized. As a result, the heterogeneity in both task requirements and computing resources poses a challenge in designing an efficient task offloading method. 

\par To overcome the above challenges, we propose a joint optimization approach for computing resource allocation and task offloading. The contributions are as follows.
\begin{itemize}
	\item \textit{\textbf{System Architecture.}} We propose a hierarchical VFC architecture consisting of a vehicle layer with a set of task vehicles (TVs), a fog layer with a set of FVs, an edge layer with a set of RSUs, and a control layer with a macro base station (MBS). Under the coordination of the MBS, the task offloading decisions of vehicles and the computing resource allocation decisions of FVs and RSUs are determined.

    \item \textit{\textbf{Problem Formulation.}} Considering the delay sensitivity of the vehicular tasks, we formulate a delay minimization optimization problem (DMOP) to minimize the task completion delay of vehicles under the energy constraints of TVs, RSUs, and FVs. Moreover, we prove that DMOP is an NP-hard mixed integer nonlinear programming (MINLP) problem.
    
    \item \textit{\textbf{Algorithm Design.}} To solve the DMOP, we propose a joint computing resource allocation and task offloading approach (JCRATOA), which includes the components of computing resource allocation and task offloading. Specifically, for computing resource allocation, the problem is decomposed into subproblems of RSU computing resource allocation and FV computing resource allocation, which are solved by using a convex optimization-based method and a contract theory-based incentive mechanism, respectively. For task offloading, we present a two-sided matching method by employing the matching game. The proposed JCRATOA offers a suboptimal solution with acceptable computational complexity, thus ensuring a balance between solution quality and efficiency.

    \item \textit{\textbf{Performance Evaluation.}} The performance of the proposed JCRATOA is evaluated through theoretical analysis and simulation. First, we prove that the worst-case computational complexity of JCRATOA is polynomial. Moreover, the simulation results demonstrate that the proposed JCRATOA clearly outperforms the other benchmark approaches in terms of {\color{b1}task completion delay, task completion ratio, system throughput, and resource utilization fairness}, while ensuring the energy constraints of TVs, RSUs, and FVs.

\end{itemize}

\par The remainder of this paper is organized as follows. Section \ref{Sec_Related_Work} reviews the related work. In Section \ref{System_Model}, we present the system model. Next, the optimization problem is formulated and analyzed in Section \ref{sec_Problem_Formulation}. The proposed JCRATOA is presented in Section \ref{The_Proposed_Joint_Optimization_Approach}. \textcolor{b1}{Sections \ref{SIMULATION_RESULTS_AND_ANALYSIS} and \ref{discussion} show the simulation results and discussions.} This work is concluded in Section \ref{CONCLUSION}. 

\section{Related Work}
\label{Sec_Related_Work}
 
\textcolor{b}{In this section, we comprehensively review the existing research works. Moreover, we summarize the differences between the related works and this work in Table 1 of the supplementary material.}

\subsection{Edge-assisted Vehicular Network Architecture}

\par MEC has been extensively studied to extend the computing capability of the vehicles. For example, Shah \textit{et al.}~\cite{Shah2022SDN} considered a software-defined networking-based MEC architecture for vehicular networks, where multiple MEC-enabled RSUs provide computing services for vehicles. Moreover, Li \textit{et al.}~\cite{Li2024Mob} introduced a non-orthogonal multiple access (NOMA)-assisted vehicular framework, where an RSU equipped with an MEC server offers computation service for vehicles on the road segment. Additionally, Jung \textit{et al.}~\cite{Jung2024GamiCO} presented a multi-interface and MEC-enabled vehicular architecture with multiple mmWave-based small base stations and a cellular-based macro base station. {\color{b}Furthermore, Sun \textit{et al.}~\cite{Sun2024Profit} considered an MEC-enabled cooperative vehicular networking architecture, where vehicles communicate with base stations via cellular networks and can offload computation tasks to MEC servers. Besides, Wang \textit{et al.}~\cite{Wang2025Delay} considered a cell free massive multiple input multiple output enabled VFC network. In this system, vehicles communicate with nearby RSUs to offload computation tasks for cooperative processing.} However, the abovementioned studies mainly focused on the MEC-enabled vehicular system, but do not fully leverage the computing capabilities of the vehicles. Due to the high deployment costs, it is unrealistic to densely deploy MEC servers. As a result, the amount of computing requirements could lead to overloads at MEC servers and long delay, especially during the peak periods.

\par To alleviate the workloads at MEC servers and reduce processing delays, VFC has emerged as a promising solution by leveraging the underutilized computing resources of nearby vehicles. For example, Lin \textit{et al.}~\cite{Lin2024Energy} proposed a multi-fog-assisted VFC system to support inter-vehicular task offloading. Furthermore, Wei \textit{et al.}~\cite{Wei2023OCVC} presented a cooperative VFC architecture, where each vehicle can join different fogs simultaneously. This architecture allows the computing resources to be exploited in an overlapping manner. Besides, Zhang \textit{et al.}~\cite{Zhang2022Joint} employed a collaborative VFC framework, where a single MBS, a set of edge servers, and the vehicles with abundant computing resources cooperatively provide services for the vehicles with limited computation powers. Additionally, Mao \textit{et al.}~\cite{Mao2023Demand} proposed an on-demand capacity planning VFC system, where the FVs are routed to the places with computing demands. {\color{b}Moreover, Yin \textit{et al.}~\cite{Yin2024Joint} proposed a hybrid offloading vehicle edge computing system, where vehicles can offload computational tasks to RSUs or other vehicles.}

\par However, the aforementioned studies were conducted under the assumption of symmetric information between the controllers and FVs, which indicates that all vehicles are willing to act as FVs. Nevertheless, in a realistic VFC system, vehicles are often reluctant to share information due to the selfishness and privacy sensitivity, which leads to asymmetric information between the controllers and vehicles. To address the limitations of existing works, we propose a hierarchical VFC architecture that operates under asymmetric information. This architecture effectively integrates the near-computing capabilities of RSUs and the idle computing resources of FVs, while accounting for the characteristics of privacy awareness and selfishness exhibited by FVs.

\subsection{Resource Allocation and Task Offloading}

\par Researchers have studied various aspects of VFC systems, with a primary focus on resource allocation and task offloading. Given the delay sensitivity of vehicular tasks, several studies focused on delay minimization for VFC. For example, Tang \textit{et al.}~\cite{Tang2020Mobile} focused on minimizing the total response latency of VFC by jointly optimizing the task scheduling and resource allocation. Furthermore, Nan \textit{et al.}~\cite{Nan2023Joint} aimed to minimize the average latency of task offloading through optimizing the task offloading and computational resource allocation in VFC. Moreover, Fan \textit{et al.}~\cite{Fan2023Joint} formulated a joint resource allocation and task offloading problem for VFC, with the aim of minimizing the total task processing delay for all vehicles. Hou \textit{et al.}~\cite{Hou24Hierarchical} aimed to minimize the mean offloading delay of tasks in VFC by optimizing the computing resource allocation and task offloading. {\color{b}Besides, Hu \textit{et al.}~\cite{Hu2025Computation} jointly optimized the offloading decisions, computing resource allocation, and transmission power allocation to minimize the maximum service delay experienced by all vehicles.} The abovementioned works mainly focused on minimizing latency, without considering the impact of energy consumption on system performance. However, different from the cloud computing, the MEC servers and FVs have limited energy resources. Prioritizing delay optimization alone can lead to a significant increase in energy consumption for RSUs and FVs, which is impractical for the VFC system.
 
\par Considering the energy constraints of RSUs and FVs, several studies took into account the energy consumption in the problem formulation. For example, Cong \textit{et al.}~\cite{Cong2023Latency} explored the problem of minimizing the task offloading cost in vehicular networks, where the cost was theoretically modeled by integrating the delay and energy consumption. Moreover, Zhang \textit{et al.}~\cite{Zhang2024Partial} studied the resource allocation strategy for a multi-user VFC system, with the aim of minimizing the weighted sum of delay and energy consumption. Furthermore, Huang \textit{et al.}~\cite{Huang2022Revenue} aimed to reduce the energy consumption of task execution for vehicles under the constraints of delay. {\color{b}Additionally, Tian \textit{et al.}~\cite{Tian2025Task} investigated the task offloading and resource allocation in vehicular edge computing networks, taking both energy consumption and delay into consideration. Besides, Wakgra \textit{et al.}~\cite{Wakgra2024Multi} considered an optimization problem of task offloading for VFC system, with the aim of minimizing the average weighted sum cost of the system in terms of delay and energy consumption.} However, these works did not consider the key dynamic features of the VFC system such as the mobility of vehicles and the variability of the wireless channel, which have a significant impact on decision making. In contrast to these studies, we formulate a delay minimization problem under the energy constraints of both RSUs and FVs, while also considering the channel dynamic and vehicle mobility.

\subsection{Optimization Approaches}

\par To solve the complex optimization problem of resource allocation and task offloading, researchers have explored various optimization approaches by adopting advanced methods such as heuristic algorithms, and deep reinforcement learning (DRL). For example, Sun \textit{et al.}~\cite{Sun2023Vehicular} designed an ant colony algorithm for the multi-objective optimization of task offloading and job scheduling in the vehicular edge computing networks. Wang \textit{et al.}~\cite{Wang2022Distributed} developed an online heuristic algorithm to make real-time offloading decisions for vehicles within the VFC system. Moreover, Huang \textit{et al.}~\cite{Huang2023Joint} proposed a dynamic task offloading and resource allocation approach by leveraging DRL to deal with the high-dimensional and continuous states and the action spaces. \textcolor{b}{Luo \textit{et al.}~\cite{Luo2024DRL} proposed a DRL algorithm with embedded penalty mechanisms to find out real-time solution for computational resource optimization of MEC servers. Furthermore, Liu \textit{et al.}~\cite{Liu2024Joint} presented a DRL-based dual timescale scheme to jointly optimize the long-term service caching and short-term  offloading and resource allocation. In~\cite{10736570}, the authors proposed a diffusion-based DRL approach for deep neural network task offloading, and resource allocation in vehicular networks. Additionally, Hazarika \textit{et al.}~\cite{Hazarika2024} explored a federated DRL approach for efficient learning while maintaining privacy in vehicular networks. Besides, Shang \textit{et al.}~\cite{Shang2024Joint} designed a proximal policy optimization (PPO)-based approach to jointly optimize service caching and task offloading for mobile edge-cloud computing. He \textit{et al.}~\cite{He2025Low} proposed a multi-objective task-aware service offloading algorithm for medical Internet of things systems by employing deep deterministic policy gradients (DDPG).}

\par However, the aforementioned approaches may not be suitable for our VFC system due to several limitations. First, swarm intelligence algorithms generally require numerous iterations to converge, thereby making them less adaptable to the dynamic nature of VFC systems. Moreover, heuristic algorithms often fail to guarantee optimal solutions and tend to require significant iterations, which results in high computational overhead and longer processing delays. In addition, while DRL is effective in training agents to make decisions, it generally requires extensive sample data to achieve the optimal outcomes, which results in long training times and considerable computational resources. This makes it unsuitable to solve the joint optimization problem in the delay-sensitive and resource-limited VFC system.


\par Considering that FVs are selfish and are unwilling to share the idle resources, recent studies have focused on designing incentive methods. For example, Sun \textit{et al.}~\cite{Sun2022Dynamic} proposed a two-stage incentive mechanism based on the Stackelberg game. This mechanism enables the interaction of vehicles and RSUs for efficient resource allocation. Additionally, Cao \textit{et al.}~\cite{Cao2024Incentive} proposed an optimal differentiated pricing method to stimulate the service vehicles to allocate the available computing resources to the task vehicles. Moreover, Dai \textit{et al.}~\cite{Dai2021Vehicle} modeled the trading process between UAVs and vehicles as a bargaining game to incentivize vehicles for task offloading. \textcolor{b}{Besides, Zhang \textit{et al.}~\cite{Zhang2025Resource} proposed a multi-task incentive mechanism through optimizing reward rates. Chen \textit{et al.}~\cite{Chen2024Game} presented a price incentive mechanism to motivate idle vehicles to participate in the task offloading process.} However, the pricing strategies in the aforementioned works were developed under the assumption of symmetric information, which neglects the self-interested nature and privacy concerns of vehicles. In the realistic VFC system, the information is often asymmetric, where the selfish FVs may misreport their actual states, such as the amount of computational resources. 

	\begin{table*}[!t]
        \centering
	\vspace{0em}
	\setlength{\abovecaptionskip}{0pt}%
	\setlength{\belowcaptionskip}{0pt}%
	\caption{Summary of notations}
	\label{tab_notation}
	\renewcommand*{\arraystretch}{.9}
	\begin{center}
		\begin{tabular}{|p{.187\textwidth}|m{.32\textwidth}||m{.1\textwidth}|m{.28\textwidth}|}
			\hline
			\textbf{Symbol}&\textbf{Description}&\textbf{Symbol}&\textbf{Description}\\
			\hline
				$\mathcal{N} = \{1, 2, \ldots, n, \ldots, N\}$ & The set of TVs&$v\in\mathcal{N}\cup\mathcal{M}$&The index of vehicle v\\ 
			\hline
				$\mathcal{M} = \{1, 2, \ldots, m, \ldots, M\}$ & The set of FVs&$s\in\mathcal{K}\cup\mathcal{M}$ & The index of edge server $s$ \\
			\hline
				$\mathcal{K} = \{1, 2, \ldots, k, \ldots, K\}$ & The set of RSUs&$\Psi_n(t)$ & The task generated by TV $n$ at time t \\ 
			\hline
                    $\mathcal{T}=\{1, 2, \ldots,t,\ldots,T\}$ & System timeline &$D_n^{\text{in}}(t)$ & The input data size of the task\\ 
                    
                \hline
                    $D_n^{\text{out}}(t)$ & The output data size of the task & $C_n(t)$ & The required computing resources of the task \\ 
                    
                \hline
                    $t^{\max}_n(t)$& The maximum allowed delay for task completion & $\textbf{v}_{v}(t)$ &The velocity vector at time slot $t$\\ 
                    
                \hline
                    $\bar{\textbf{v}}_{v}$,$\alpha$& The asymptotic mean and memory level of velocity, & $\mathbf{w}_v$ & Uncorrelated random Gaussian process \\ 
                    
                \hline
                    $\mathbf{q}_{v}(t)=[x_v(t),y_v(t)]$& The horizontal coordinate of vehicle $v$ & $r_{n,k}(t)$ & The uplink data rate from TV $n$ to RSU $k$ \\ 
                    
                \hline
                    $B_{n,k}$&The communication bandwidth between TV $n$ and RSU $k$& $p_{n,k}$ & The transmit power from TV $n$ to RSU $k$\\ 
                \hline
                    $\gamma_{n,k}(t)$& The channel gain between TV $n$ and RSU $k$ & $N_0$ & Noise power \\ 
                \hline
                    $d_{n,k}(t)$&The distance between TV $n$ and RSU $k$& $\alpha_k$ & The path loss exponent of the V2I link \\ 
                \hline
                    $h_{n,k}(t)$ & The component of small-scale fading of the V2I link& $r_{n,m}(t)$ & The data rate from TV $n$ to FV $m$ \\ 
                \hline
                    $B_{n,m}$& The communication bandwidth between TV $n$ and the FV $m$ &  $p_{n,m}$ & The transmit power from TV $n$ to FV $m$\\ 
                \hline
                    $\gamma_{n,m}(t)$& The channel gain between TV $n$ and the FV $m$ & $d_{n,m}(t)$ & The distance between TV $n$ and FV $m$ \\ 
                \hline
                    $\alpha_m$& The path loss exponent of the V2V link & $h_{n,m}(t)$ & The component of small-scale fading of the V2V link \\ 
                \hline  
                    $o_{n,a}(t),a\in n\cup \mathcal{M}_n(t) \cup \mathcal{K}$&   The task offloading decision& $\mathcal{M}_n(t)$ & The set of FVs within the range of TV $n$ in time slot $t$ \\ 
                \hline  
                    $o_{n,n}(t)$/$o_{n,k}(t)$/$o_{n,m}(t)$& The task offloading decision for offloading task locally/RSU $k$/FV $m$ & $o_{n,s}(t)$ & The task offloading decision for offloading task on edge server $s$ \\                     
                \hline  
                    $T_{n,n}(t)$/$T_{n,k}(t)$/$T_{n,m}(t)$& The task completion delay for processing task on vehicle $n$/RSU $k$/ FV $m$ & $f_n$ & The computing resources of TV $n$ \\ 
                \hline  
                    $f_{n,k}(t)$& The computing resources allocated by the RSU $k$ to task $\Psi_n(t)$ in time slot $t$ & $r_f$ &The data rate of fiber link \\ 
                \hline  
                    $f_{n,m}(t)$& The computing resources allocated by the FV $m$ to task $\Psi_n(t)$ in time slot $t$ &$T_{n}(t)$  & Total completion delay \\ 
                \hline   
                    $E_{n}(t)/E_{n,s}(t)$& The energy consumption of TV $n$/server $s$  &  $\kappa^{\text{TV}}/\kappa_s$&The effective switched capacitance of the TV $n$/server $s$ \\ 
                \hline 
                    $\mathbf{O},\mathbf{F}$& The decisions of task offloading and  computing resource allocation& $e$ & The unit cost of energy consumption \\ 
                \hline 
                    $E_n^{\text{max}}$/$E_k^{\text{max}}$/$E_m^{\text{max}}$& The energy constraints of TV $n$/RSU $k$/FV $m$ & $f_{k}^{\text{max}}$/$f_{m}^{\text{max}}$ & The maximum computing resources of RSU $k$/FV $m$ \\ 
                \hline 
                    $\sigma_l$& The strength of the willingness to contribute resources  & $f_l^{\text{max}}$ & The maximum computational resource that an FV can contribute \\ 
                \hline 
                    $\Theta=\{\theta_1,\theta_2,\ldots ,\theta_L\}$& The set of types of FVs & $\Pi_{n}(t)$ & Matching result \\ 
                \hline 
                    $f_l(t)$& The computing resources allocated by FV with type $\theta_l$  &$w_l(t)$  & The rewards of FV with type $\theta_l$ \\ 
                \hline 
                    $E_l(t)$& The energy consumed by the FV with type $\theta_l$ for task computing & $M_l$ & The total number of type $\theta_l$ FVs \\ 
                \hline 
                    $w_l^*(t)$,$f_l^*(t)$& The optimal rewards and computing resource allocation of FV with type $\theta_l$& $f_{n,m}^*(t)$ & The optimal computing resources that each FV $m$ should allocate to TV $n$ \\ 
                \hline 
                    $\mathcal{P}_{n}(t)/\mathcal{P}_{s}(t)$& The preference lists of TVs/servers & $(\mathcal{A},\mathcal{P}(t),\Pi(t))$ & Current matching \\ 
                \hline       
                    $\Phi_{n,s}(t)/\Phi_{s,n}(t)$& The preference value of TV $n$/server $s$ on server $s$/TV $n$  & $\tau$ & Time slot duration \\ 
                \hline   
		\end{tabular}
	\end{center}
    \vspace{-2em}
\end{table*}     

\section{Models and Preliminaries}
\label{System_Model}

\par In this section, we first propose a hierarchical VFC architecture. Then, we introduce the basic models, communication model, and computation model in the VFC system. The notations are listed in Table \ref{tab_notation}.

\subsection{System Overview}
\par In Fig. \ref{system_model}, we consider a hierarchical VFC architecture under asymmetric information in urban scenario. This architecture comprises a vehicle layer with a set of TVs $\mathcal{N} = \{1, 2, \ldots, n, \ldots, N\}$ and a set of FVs $\mathcal{M} = \{1, 2, \ldots, m, \ldots, M\}$, an edge layer with a set of RSUs $\mathcal{K} = \{1, 2, \ldots, k, \ldots, K\}$, and a control layer with an MBS. Specifically, \textit{at the vehicle layer}, the TVs periodically generate vehicular tasks such as autonomous driving and infotainment applications, which require offloading services due to their delay sensitivity and computing intensity. Moreover, each TV can decide to process the task locally, upload it to the connected RSU, or offload it to an FV within its range. Additionally, the FVs share the idle computing resources to the nearby TVs for task processing. \textit{At the edge layer}, the RSUs equipped with the MEC servers\footnote{The RSU and MEC server will be used interchangeably.} are deployed along the road with non-overlapping coverage radius to provide offloading services for the TVs. These RSUs are interconnected with the MBS and with each other through fiber links~\cite{Sun2024TJCCT}. In addition, each RSU is responsible for collecting local information on its own status, the vehicle states, and the channel state information, which are then uploaded to the control layer. \textit{At the control layer}, the MBS is equipped with a controller for decision making, and it is connected to the RSUs for information collection and decision distribution. 

\par In the open and dynamic VFC system, the FVs are often reluctant to disclose their private information (e.g., real-time resource availability or collaboration intent) due to privacy concerns, thus leading to information asymmetry between the MBS and FVs. Consequently, this asymmetric information prevents the MBS from obtaining comprehensive knowledge on the available resources and collaboration intentions of the FVs. As a result, without accurate and complete knowledge of the system states, the MBS struggles to optimize resource allocation and manage workloads effectively, which ultimately results in inefficiencies in utilizing the idle computational resources of FVs.

\par The system operates in a time-slotted manner, where the system time is discretized into $T$ time slots $\mathcal{T}=\{1, 2, \ldots,t,\ldots,T\}$ with equal slot duration $\tau$~\cite{Zhang2024MultiObjective}. Note that the TVs and FVs are collectively referred to as vehicles, indexed by $v\in\mathcal{N}\cup\mathcal{M}$, and the RSUs and FVs are collectively referred to as edge servers, indexed by $s\in\mathcal{K}\cup\mathcal{M}$.

\begin{figure}[t] 
	\setlength{\abovecaptionskip}{0pt}   
	\setlength{\belowcaptionskip}{0pt} 
	\centering
	\includegraphics[width =3.5in]{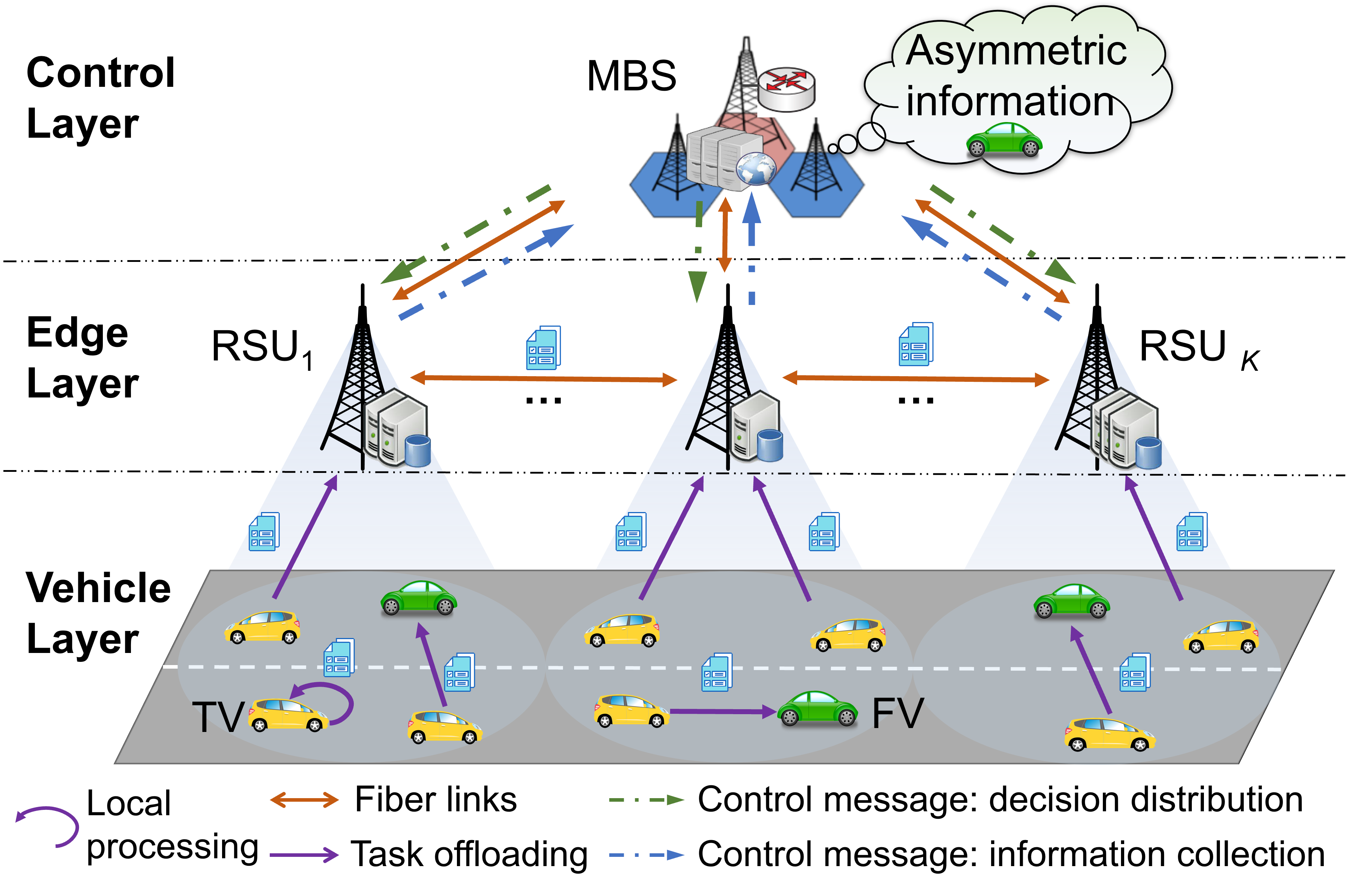}
        \caption{\textcolor{b}{The architecture of the hierarchical VFC system under asymmetric information consists of a vehicle layer, an edge layer, and a control layer. Each TV can execute the tasks locally, or offload the tasks to an RSU or an FV for edge computing. The RSUs, equipped with MEC servers, provide edge computing services and are connected to the MBS via fiber links. The MBS acts as a centralized controller, responsible for collecting system information and making offloading decisions under incomplete knowledge caused by the private information of FVs.}}
	\label{system_model}
\end{figure}

\subsection{Basic Models}

\par The basic models of the system are given as follows.

\par\textbf{Task Model.} We consider that each TV generates a computational task per time slot~\cite{Wu2021Fog}. Specifically, the task of TV $n$ is denoted by $\Psi_n(t)=(D_n^{\text{in}}(t),D_n^{\text{out}}(t),C_n(t),t^{\max}_n(t))$, where $D_n^{\text{in}}(t)$ represents the input data size, $D_n^{\text{out}}(t)$ denotes the output data size, $C_n(t)$ is the required computing resources of the task (in cycles), and $t^{\max}_n(t)$ indicates the maximum allowed delay for task completion.

\par \textbf{Vehicle Mobility Model.} The mobility of each vehicle $v$ is modeled as a Gauss-Markov mobility model~\cite{Sun2025Aerial}. Specifically, the velocity of vehicle $v$ is given as follows:
\begin{equation}
    \textbf{v}_{v}(t+1)=\alpha \mathbf{v}_{v}(t)+(1-\alpha)\bar{\textbf{v}}_{v}+\sqrt{1-\alpha^2}\mathbf{w}_v, \ v\in\mathcal{N}\cup\mathcal{M},
\end{equation}

\noindent where $\textbf{v}_{v}(t)$ denotes the velocity vector at time slot $t$, $\bar{\textbf{v}}_{v}$ is the asymptotic mean of velocity, and $\alpha$ ($0\leq\alpha\leq1$) denotes the memory level, which reflects the temporal-dependent degree. \textcolor{b}{Moreover, $\mathbf{w}_v$ represents the uncorrelated random Gaussian process, i.e., $\mathbf{w}_v \sim{f^{\text{Gua}}(0,\varsigma^2)}$, where $\varsigma$ denotes the asymptotic standard deviation of velocity.} We denote the horizontal coordinate of each vehicle $v$ as $\mathbf{q}_{v}(t)=[x_v(t),y_v(t)]^{\text{T}}$. Therefore, the location of each vehicle $v$ evolves as:
\begin{equation}
    \mathbf{q}_{v}(t+1)=\mathbf{q}_{v}(t)+\mathbf{v}_{v}(t)\tau, \ v\in\mathcal{N}\cup\mathcal{M}.
\end{equation}

\subsection{Communication Model}

\par To mitigate the unreliable communication caused by interference, we consider that each server $s$ utilizes different frequency band to provide computing services for TVs. Specifically, the task of a TV can be offloaded to an RSU via V2I communication links, and to an FV through V2V communication links. Moreover, we consider that each RSU can serve multiple TVs in each time slot due to the relatively powerful computing capability, while each FV can only serve one TV per time slot because of its limited resources.
Considering the complexity of the communications in vehicular networks, the channel gain is calculated by integrating the commonly used probabilistic LoS channel with the large-scale and small-scale fadings as
\begin{equation}
\label{eq_channelPowerGain}
	h_{n,s}^t={\mathbb{P}_{n,s}^\text{L}(t)} h_{n,s}^{t,\text{L}}(t) + (1-\mathbb{P}_{n,s}^\text{L}(t))h_{n,s}^{t,\text{N}}(t),
\end{equation}

\noindent where $\mathbb{P}_{n,s}^\text{L}(t)$ denotes the probability of LoS transmission between TV $n$ and edge server $s$, $h_{n,s}^{x}(t)$ represents the channel power gain between TV $n$ and edge server $s$, and $ x\in\{\mathrm{L},\mathrm{N}\}$ represents LoS or NLoS links. Moreover, the details of $\mathbb{P}_{n,s}(t)$ and $h_{n,s}^{x}(t)$ are presented as follows.

\subsubsection{LoS Probability}

\par For V2I communication, according to the 3GPP standard \cite{3GPPTR389012020}, the LoS and NLoS probabilities of the communication between TV $n$ and RSU $s$ (i.e., $s\in\mathcal{K}$) is given as:
\begin{equation}
\label{LoS_V2I}
    \begin{aligned}
    \mathbb{P}_{n,s}^\text{L}(t)=\left\{\begin{array}{cl}
1, \, d_{n,s}^h(t)\leq \, 18 \ \text{m} \\
\frac{18}{d_{n,s}^h(t)}+e^{\frac{-d_{n,s}^h(t)}{36}}\big(1-\frac{18}{d_{n,s}^h(t)}\big), \, d_{n,s}^h(t)>18 \ \text{m},
\end{array}\right.
    \end{aligned}
\end{equation}

\noindent where $d_{n,s}^h(t)$ represents the horizontal distance between TV $n$ and RSU $s$.

\par For V2V communication, according to \cite{ETSITR1032571}, the LoS probability between TV $n$ and FV $s$ ($s\in\mathcal{M}$) is given as
\begin{equation}
\label{LoS_V2V}
  \mathbb{P}_{n,s}^\text{L}(t)=\min\{1,1.05e^{-0.014d_{n,s}(t)}\}
\end{equation}

\noindent where $d_{n,s}(t)$ is the distance between TV $n$ and FV $s$.

\subsubsection {Channel Gain}

\par The channel gain between TV $n$ and server $s$ in time slot $t$ is uniformly given as~\cite{Sun2024TJCCT} $h_{n,s}^{x}(t)=|h_{n,s}^{\text{Sm},x}(t)|^2/(10^{-h_{n,s}^{\text{La},x}(t)/10})$, where $h_{n,s}^{\text{Sm},x}(t)$ and $h_{n,s}^{\text{La},x}(t)$ denote the parameters of \textit{small-scale fading} and \textit{large-scale fading}, respectively, which are given in detail as follows.

\par \textbf{Small-scale fading.} The small-scale fading between TV $n$ and server $s$ can be modeled as a parametric-scalable and good-fitting generalized fading, i.e., Nakagami-$m$ fading \cite{Sun2024TJCCT}, which is given as:
\begin{equation}
\label{eq_LoS_probability}
    \begin{aligned}
        &h_{n,s}^{\text{Sm},x}(t)\sim f^{\text{Nak}}\big(h_{n,s}^{\text{Sm},x}(t),\mathbf{m}_{y}^{x}\big) \\
        &= \frac{2{(\mathbf{m}_y^x)}^{\mathbf{m}_y^x} (h_{n,s}^{\text{Sm},x}(t))^{2\mathbf{m}_y^x-1} e^{(-\mathbf{m}_y^x (h_{n,s}^{\text{Sm},x}(t))^2/\overline{p})}}{\Gamma(\mathbf{m}_y^x) (\overline{p})^{\mathbf{m}_{y}^{x}}}, \ j\in \{b,\mathcal{U}\},
    \end{aligned}
\end{equation}

\noindent where $\overline{p}$ is the average received power, $\Gamma(\cdot)$ is the Gamma function, and $\mathbf{m}_{y}^{x}\in\{m_\mathrm{V2I}^\mathrm{L},m_\mathrm{V2I}^\mathrm{N},m_\mathrm{V2V}^\mathrm{L},m_\mathrm{V2V}^\mathrm{N}\}$ is the Nakagami-$m$ fading parameters of LoS/NLoS channel for V2I/V2V communication. 

\par \textbf{Large-scale fading.} First, the large-scale fading of LoS link for V2I communication between TV $n$ and RSU $k$ is given as~\cite{3GPPTR389012020}:
\begin{sequation}  
\label{LoS_V2I}  
    h_{n,k}^{\text{La},\text{L}}(t)=  
    \left\{  
    \begin{aligned}   
    &32.4 + 21\log_{10} (d_{n,k}(t)) + 20\log_{10}(f_c)+\vartheta^{\text{L}}, \\& 10\leq d_{n,k}^h(t)\leq d_{n,k}^{\prime}, \\   
    &32.4+40 \log_{10}(d_{n,k}(t))+20 \log_{10}(f_c)-9.5 \\&\times \log_{10}\big((d_{n,k}^{\prime})^2+(H_n-H_k)^2\big)+\vartheta^{\text{L}}, \\ &d_{n,k}^{\prime}<d_{n,k}^h(t) \leq 5 \ \text{km},
    \end{aligned}  
    \right.  
\end{sequation}  

\noindent where $f_c$ denotes the center radio frequency (in Hz), $d_{n,k}^h(t)$ represents the horizontal distance between TV $n$ and RSU $k$, and $\vartheta^{\text{L}}$ is the shadow fading. Moreover, $d_{n,k}^{\prime}=4H_kH_nf_c/c$ is the breakpoint distance, where $c=10\times 10^8$ m/s denotes the light speed, $H_k$ represents the effective antenna height at RSU $k$, and $H_n$ is the effective antenna height at TV $n$.

\par Second, the large-scale fading of NLoS link for V2I communication is given as~\cite{3GPPTR389012020}:
\begin{sequation}  
\label{NLoS_V2I}    
    \begin{aligned}
         &h_{n,k}^{\text{La},\text{N}}(t)= \max\big( h_{n,k}^{\text{La},\text{L}}(t), 35.3\log_{10}(d_{n,k}(t))+22.4\\&+21.3\log_{10}(f_c)-0.3(H_n-1.5)\big).
    \end{aligned}  
\end{sequation}  

\par Third, the large-scale fading of LoS link for V2V communication is given as~\cite{ETSITR1032571}:
\begin{equation}  
\label{LoS_V2V}  
\begin{aligned}
    h_{n,s}^{\text{La},\text{N}}(t)= 38.77+16.7\log_{10}(d_{n,k}(t))+18.2\log_{10}(f_c)
\end{aligned}
\end{equation}  

\par Finally, the large-scale fading of NLoS link for V2V communication is given as~\cite{ETSITR1032571}:
\begin{equation}  
\label{NLoS_V2V}  
\begin{aligned}
    h_{n,s}^{\text{La},\text{N}}(t)= 36.85 + 30\log10(d_{n,k}(t)) + 18.9\log_{10}(f_c).
\end{aligned}
\end{equation}  

\subsubsection{\textcolor{b}{Transmission Rate}}

\par {\color{b} For V2I and V2V communications, we adopt the orthogonal frequency-division multiple access (OFDMA) technique, which has been widely used in latency-sensitive and resource-constrained MEC systems.} Therefore, the transmission rate from TV $n$ to server $s$ is given as:
\begin{equation}
r_{n,s}(t)=B_{n,k}\log_2\big(1+p_{n}h_{n,k}(t)/N_0\big), \, \forall n\in \mathcal{N}, s\in\mathcal{K} \cup \mathcal{M},
\end{equation}

\noindent where $B_{n,s}$ denotes the communication bandwidth between TV $n$ and 
 edge server $s$, $p_{n}$ represents the transmit power of TV $n$, $N_0$ is the background noise, and $h_{n,s}(t)$ means the channel gain.

\begin{remark}
 Although OFDMA does not support spectrum reuse, its orthogonal subcarriers eliminate mutual interference, leading to more reliable and faster transmission. Moreover, the advanced multiple access schemes such as non-orthogonal multiple access (NOMA) requires dynamic user grouping and successive interference cancellation~\cite{Liu2024RIS}, which introduces high complexity for RSUs and FVs with limited processing capability. Therefore, OFDMA is more practical for the delay-sensitive and resource-constrained VFC scenario. 
\end{remark}

\subsection{Computation Model}

\par The tasks generated by each TV $n$ can be computed locally or offloaded to FVs and RSUs, which depends on the task offloading decision. Specifically, the task offloading decision of TV $n$ at time slot $t$ is defined as $o_{n,a}(t)\in\{0,1\}$, where $a\in n\cup \mathcal{M}_n(t) \cup \mathcal{K}$ denotes the offloading destinations of TV $n$ and $\mathcal{M}_n(t)$ is the set of FVs within the range of TV $n$. Moreover, $o_{n,n}(t)=1$ denotes that the task is processed locally, $o_{n,k}(t)=1$ means that the task is offloaded to RSU $k$, and $o_{n,m}(t)=1$ indicates that the task is offloaded to FV $m$. \textcolor{b}{Note that the delay of result feedback can be disregarded when considering the task completion delay. This is because for many intelligent applications, the size of the results is typically significantly smaller than that of the input data.}

\subsubsection{Task Completion Delay}

\par When TV $n$ processes task $\Psi_n(t)$ locally, the task completion delay is given as
\begin{equation}
\label{eq_local_delay}
    T_{n,n}(t)=C_n(t)/f_{n},
\end{equation}

\noindent where $f_n$ represents the computing resources of TV $n$.

\par When TV $n$ offloads task $\Psi_n(t)$ to RSU $k$, we consider two cases. In the first case, if TV $n$ is located in the coverage of the RSU, the task is transmitted directly to RSU $k$ and executed there. In the second case, the task is first transmitted to the nearest RSU $k^{\prime}$ from TV $n$, and then forwarded to RSU $k$ via fiber links. Therefore, the task completion delay primarily consists of task upload delay, task relay delay, and task computation delay, which is given as:
\begin{equation}
\label{eq_off_delay_V2I}
 T_{n,k}(t)=\underbrace{D_n^{\text{in}}(t)/r_{n,k^{\prime}}(t)}_{\text{Upload delay}}+\underbrace{h_{k^{\prime},k}C_n(t)/r_f}_{\text{Relay delay}}+\underbrace{C_n(t)/f_{n,k}(t)}_{\text{Computation delay}},
\end{equation}

\noindent where $k^{\prime}$ represents the RSU within whose coverage TV $n$ is located, $h_{k,k^{\prime}}$ represents the number of hops between RSU $k^{\prime}$ and RSU $k$, $r_f$ denotes the data rate of fiber link, and $f_{n,k}(t)$ is the computing resources allocated by  RSU $k$ to task $\Psi_n(t)$ in time slot $t$. Note that $h_{k,k^{\prime}}=0$ if $k= k^{\prime}$.

\par When TV $n$ offloads task $\Psi_n(t)$ to FV $m$ ($m\in\mathcal{M}_n(t)$), the task completion delay mainly includes task upload delay and task computation delay~\cite{Sun2024JointTask}, which can be expressed as: 
\begin{equation}
\label{eq_off_delay_V2V}
 T_{n,m}(t)=\underbrace{D_n^{\text{in}}(t)/r_{n,m}(t)}_{\text{Upload delay}}+\underbrace{C_n(t)/f_{n,m}(t)}_{\text{Computation delay}},
\end{equation}

\noindent where $f_{n,m}(t)$ represents the computing resources allocated by the FV $m$ to task $\Psi_n(t)$ in time slot $t$.

\par Therefore, based on \eqref{eq_local_delay}, \eqref{eq_off_delay_V2I}, and \eqref{eq_off_delay_V2V}, the task completion delay of TV $n$ is given as:
\begin{equation}
\label{eq_total_delay}
 T_{n}(t)=o_{n,n}(t)T_{n,n}(t)+\sum_{s\in\mathcal{K}\cup\mathcal{M}_n(t)}o_{n,s}(t)T_{n,s}(t),
\end{equation}

\noindent where $s$ represents the set of potential servers (i.e., RSUs and FVs) that can provide computing service for TV $n$.

\subsubsection{Energy Consumption}

\par The energy consumption of TV $n$ includes the computation energy and transmission energy, which is given as
\begin{equation}
E_{n}(t)=\underbrace{o_{n,n}(t)\kappa^{\text{TV}} C_n(t)f_{n}^2}_{\text{Computation energy}}+\underbrace{o_{n,s}(t)p_{n,s}D_n^{\text{in}}(t)/r_{n,s}(t)}_{\text{Transmission energy}},
\end{equation}

\noindent where $s\in\mathcal{K}\cup\mathcal{M}_n(t)$, and $\kappa^{\text{TV}} \geq0$ is the effective switched capacitance for the CPU of the TV.

\par Similarly, the energy consumption of server $s$ is mainly incurred by the task computation, which can be given as:
\begin{equation}
E_{n,s}(t)=\kappa_s C_n(t)f_{n,s}^2(t), \ \kappa_s\in\{\kappa^{\text{FV}},\kappa^{\text{RSU}}\}
\end{equation}

\noindent where $\kappa_s$ is the effective switching capacitance of the CPU for server $s$. Specifically, $\kappa_s = \kappa^{\text{FV}}$ when $s\in\mathcal{M}$, and $\kappa_s = \kappa^{\text{RSU}}$ when $s\in\mathcal{K}$.

\begin{figure*}[h]
    \centering
    \setlength{\abovecaptionskip}{0pt}%
    \setlength{\belowcaptionskip}{2pt}%
    \includegraphics[width =7in]{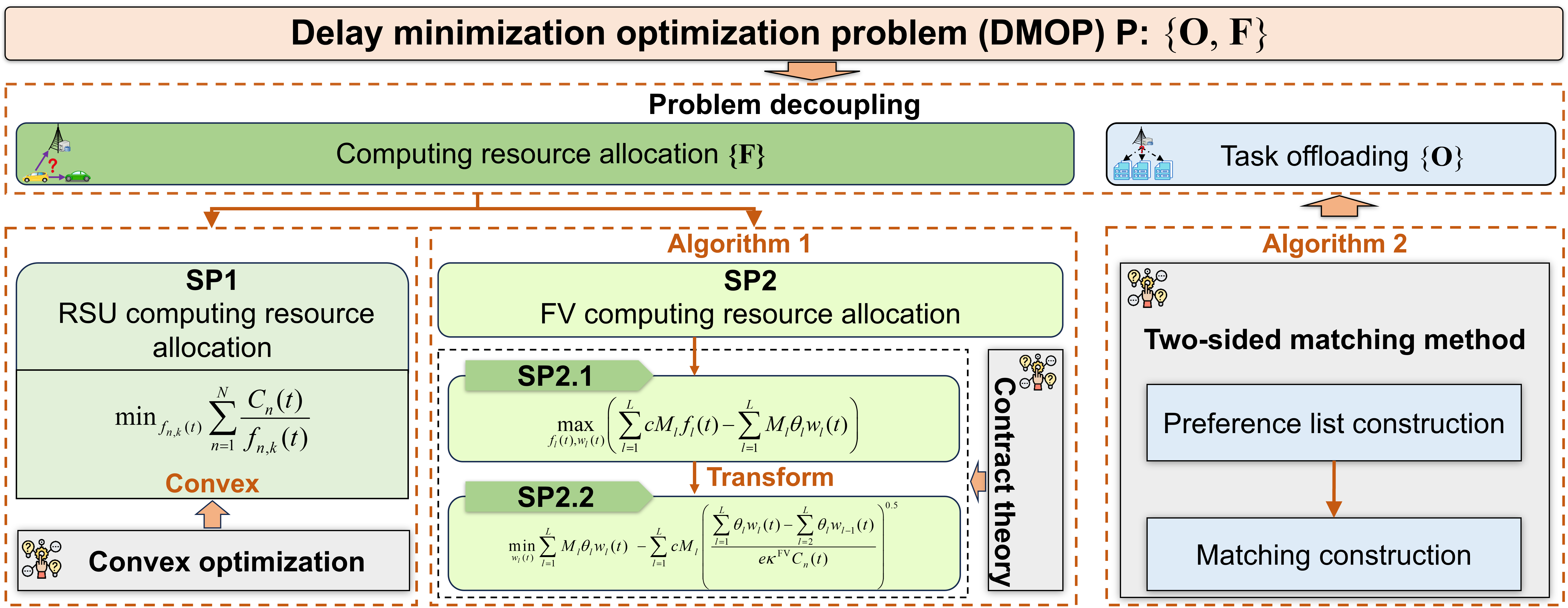}
    \caption{\textcolor{b}{The framework of JCRATOA. The original problem DMOP is first decomposed into a computing resource allocation subproblem and a task offloading subproblem. First, the computing resource allocation subproblem is further decomposed into subproblems of RSU computing resource allocation and FV computing resource allocation, which are solved by convex optimization and contract theory, respectively. Subsequently, the task offloading subproblem is solved through a two-sided matching method.}}
\label{fig_framwork}
 \vspace{-1.6em}
\end{figure*}

\section{Problem Formulation and Analysis}
\label{sec_Problem_Formulation}

\subsection{Problem Formulation}

\par \textcolor{b1}{In delay-sensitive VFC environments, ensuring the timely execution of computation-intensive and latency-critical tasks is imperative, particularly for safety-critical applications such as autonomous driving. Moreover, the limited computing capability of RSUs and FVs, coupled with frequent topology changes caused by vehicle mobility, further intensifies this delay sensitivity. In contrast, RSUs and vehicles generally possess sufficient and stable power supplies such as large batteries or direct power connections, making energy consumption less critical in the short term. Accordingly, delay is often the dominant performance metric in practical VFC systems, while energy consumption should remain within operational bounds. Therefore, rather than jointly minimizing delay and energy through a weighted-sum objective, we adopt a constraint-based formulation that minimizes delay while imposing explicit energy constraints on TVs, RSUs, and FVs. This formulation aligns with the characteristics of real-time VFC systems, providing strict guarantees on energy usage and stable delay performance without requiring complex parameter tuning \cite{Zhang2023Efficient, Huiyi2021Resource}.}

\par Consequently, the objective of this work is to minimize the task completion delay of TVs by jointly optimizing the decisions of computing resource allocation $\mathbf{F}=\{f_{n,s}(t)\}_{n\in\mathcal{N},s\in\mathcal{M}\cup\mathcal{K},t\in\mathcal{T}}$ and task offloading $\mathbf{O}=\{o_{n,a}(t)\}_{n\in\mathcal{N},a\in n\cup \mathcal{M} \cup \mathcal{K}, t\in\mathcal{T}}$ under the energy constraints. Consequently, the DMOP is formulated as
\begin{subequations}
    \begin{alignat}{2}
    \mathbf{P}: \quad &\min_{\mathbf{O},\mathbf{F}}  \sum_{n=1}^N T_n(t),\label{P}\\
    \text{s.t.} \quad 
    &o_{n,n}(t) \in \{0,1\},\forall n \in \mathcal{N},\label{P_C1}\\
    &o_{n,s}(t) \in \{0,1\},\forall n \in \mathcal{N}, \, s\in \mathcal{M}_n(t) \cup \mathcal{K}, \label{P_C11}\\
    &0 \leq o_{n,n}(t)+o_{n,m}(t)+o_{n,k}(t) \leq 1, \notag\\ 
    &\forall n \in\mathcal{N}, \, k\in \mathcal{K}, \, m\in \mathcal{M}_n(t),\label{P_C2}\\
    &T_n(t) \leq t_n^{\max}(t),\forall n \in \mathcal{N},\label{P_C3}\\
    &E_{n}(t) \leq E_n^{\text{max}},\forall n\in\mathcal{N},\label{P_C4}\\
    &\sum_{n=1}^{N}o_{n,m}(t)E_{n,m}(t) \leq E_m^{\text{max}},\forall m\in\mathcal{M},\label{P_C5}\\
    &\sum_{n=1}^{N}o_{n,k}(t)E_{n,k}(t) \leq E_k^{\text{max}},\forall k\in\mathcal{K},\label{P_C6}\\
    &\sum_{n=1}^{N}o_{n,m}(t)f_{n,m}(t) \leq f_{m}^{\text{max}},\forall  m \in\mathcal{M},\label{P_C8}\\
    &\sum_{n=1}^{N}o_{n,k}(t)f_{n,k}(t) \leq f_{k}^{\text{max}}, k \in\mathcal{K},\label{P_C9}
    \end{alignat}
\end{subequations}

\noindent where $E_n^{\text{max}}$, $E_k^{\text{max}}$, and $E_m^{\text{max}}$ represent the energy constraints of TV $n$, RSU $k$, and FV $m$, respectively. Furthermore, $f_{k}^{\text{max}}$ and $f_{m}^{\text{max}}$ denote the maximum computing resources of RSU $k$ and FV $m$. \textcolor{b}{Moreover, constraints \eqref{P_C1}, \eqref{P_C11}, and \eqref{P_C2} indicate that each TV can only select one type of task offloading decision. In other words, each TV can process its task locally, offload it to an RSU, or offload it to an FV. Furthermore, constraint (\ref{P_C3}) enforces that the task completion delay should not exceed the maximum allowable delay. Additionally, constraints (\ref{P_C4}), (\ref{P_C5}), and (\ref{P_C6}) indicate that the energy consumption of TV $n$, FV $m$, and RSU $k$ should remain within their respective energy budgets. Besides, constraints (\ref{P_C8}) and (\ref{P_C9}) guarantee that the computing resource allocation of FV $m$ and RSU $k$ do not surpass the maximum allowable resource limits.}

{\color{b}
\begin{theorem}
\label{theorem_NP_hard}
The problem formulated in DMOP is an NP-hard and non-convex MINLP.
\end{theorem}

\begin{proof} 
The proof is presented in Appendix A of the supplemental material.
\end{proof}
}

\subsection{Problem Analysis}

\par Solving DMOP directly could introduce several challenges as follows:
    \begin{itemize}
    \item \textit{MINLP problem and coupled decision variables.} \textcolor{b}{First, as presented by Theorem \ref{theorem_NP_hard}, the formulated DMOP is an NP-hard and non-convex MINLP, which is computationally intractable to solve in polynomial time.} Additionally, the decision variables of different nodes are mutual-coupled and interdependent with each other, which makes it challenging to solve the formulated DMOP directly.
    
    \item \textit{Asymmetric information.} The MBS requires detailed information about the TVs and FVs, such as the position and available resources, to make accurate decisions. However, in an open and dynamic VFC system, the privacy-aware and self-interested FVs are often reluctant to disclose the private information and voluntarily share the idle resources, thereby leading to the information asymmetry between the MBS and FVs. Consequently, the formulated DMOP becomes an optimization problem under incomplete information. This lack of complete information further increases the complexity of the problem-solving process and reduces the accuracy of the obtained solutions, particularly in achieving  efficient FV resource allocation due to the privacy concerns of the FVs.

    \item \textit{Heterogeneous preferences.} In the considered VFC system, different TVs have varying requirements for various tasks, while different RSUs and FVs possess diverse computing resources. Therefore, the TVs exhibit heterogeneous preferences on different servers for task offloading, and the servers also possess different preferences for TVs. This introduces challenges in efficiently associating each TV with an appropriate server.
    \end{itemize}

\vspace{-0.4em}

\section{The Proposed JCRATOA}
\label{The_Proposed_Joint_Optimization_Approach}

\par Based on the aforementioned challenges, achieving an optimal solution is computationally infeasible in real-time VFC scenarios due to the NP-hardness of the formulated DMOP. Therefore, we propose JCRATOA to ensure computational efficiency and practical feasibility in addressing the formulated DMOP. In this section, we first present the motivations for proposing JCRATOA. Then, we introduce JCRATOA in detail, which consists of the components of computing resource allocation and task offloading. Specifically, for computing resource allocation, the formulated DMOP is divided into subproblems of RSU computing resource allocation and FV computing resource allocation, which are solved by using the convex optimization method and an incentive mechanism, respectively. For task offloading, we present a two-sided matching method by employing the matching game. \textcolor{b}{Note that decomposing DMOP into subproblems preserves the optimality of the solution. This is because the decision variables are considered together throughout the decoupling process. Moreover, the simulation results also demonstrate the superiority of the proposed JCRATOA in terms of task processing performance under the energy constraints.} The framework of the proposed JCRATOA is given in Fig. \ref{fig_framwork}.

\subsection{Motivations}
 \par The motivations for proposing JCRATOA are presented as follows.
 
    \begin{itemize}
        \item \textit{Decoupling the interdependent decision variables.} Despite the effectiveness of DRL in decision-making, the coupled decision variables in the DMOP create complex action spaces, which leads to extensive training time and numerous interactions with the environment. Therefore, the coupling of the decision variables motivates us to decouple the DMOP into manageable subproblems. Specifically, the computing resource allocation at servers and task offloading at TVs can be naturally decoupled, as they are performed by different entities. This separation not only simplifies the decision making process, but also allows each type of node perform its respective action, which makes the solution scalable.
        
        \item \textit{Mitigating the asymmetric information.} We employ the contract theory to deal with the asymmetric information in the VFC system~\cite{Li2022Joint}. Specifically, the contract theory provides an effective framework to mitigate information asymmetry by stimulating agents to reveal private information truthfully. Moreover, the contract theory enables optimal or near-optimal resource allocation by offering rewards that motivate resource sharing. Besides, the contract theory is more scalable, as it avoids multiple rounds of communication by using the predefined contracts, which makes it suitable for dynamic and large-scale environments. 
        
        \item \textit{Handling the heterogeneous preferences.} The matching game is adopted to deal with heterogeneous preferences between TVs and servers. First, the matching game can establish mutual-beneficial matching between TVs and servers with heterogeneous preferences. This ensures that the tasks of TVs can be offloaded to satisfactory destinations. Furthermore, the matching game guarantees stable and balanced outcome, which enhances adaptivity to the dynamic VFC and prevents task overloading at certain servers. Additionally, while the matching game provides a near-optimal solution, the solution obtained using the Gale Shapley algorithm has the complexity of $\mathcal{O}(N)$, thus making it suitable for the dynamic and real-time VFC scenarios.
    \end{itemize}

\subsection{Computing Resource Allocation}
\par Considering that the RSUs and FVs are two types of servers, which are characterized by different mobility patterns and computing capabilities, we decompose the formulated DMOP into an RSU computing resource allocation subproblem and an FV computing resource allocation subproblem. Specifically, the convex optimization method is adopted to solve the subproblem of RSU computing resource allocation. Moreover, a contract theory-based incentive mechanism is proposed to motivate FVs to cooperate in resource sharing.

\subsubsection{Computing Resource Allocation of RSUs}
\label{RSU_allocation}

\par Given the task offloading decision $\mathbf{\hat{O}}=\{\hat{o}_{n,a}(t)\}_{n\in\mathcal{N},a\in n\cup \mathcal{M}_n(t) \cup \mathcal{K}, t\in\mathcal{T}}$
 and removing the irrelevant terms, problem $\mathbf{P}$ is transformed into the subproblem of RSU computing resource allocation, which is as follows:
\begin{subequations}
    \begin{alignat}{2} 
    \mathbf{SP1}: \quad &\min_{f_{n,k}(t)}\sum_{n=1}^{N}C_n(t)/f_{n,k}(t)\label{SP_RSU}\\
    \text{s.t.} \quad 
    &\eqref{P_C3}, \eqref{P_C6}, \eqref{P_C9} \notag.
    \end{alignat}
\end{subequations}

\par Problem $\mathbf{SP1}$ is a convex optimization problem, as given in Theorem \ref{the_sp1_convex}. Accordingly, problem $\mathbf{SP1}$ can be solved in polynomial time by using the Matlab fmincon tools. 

\begin{theorem}
\label{the_sp1_convex}
Problem $\mathbf{SP1}$ is a convex optimization problem.
\end{theorem}

\begin{proof}
The proof is presented in Appendix B of the supplemental material.
\end{proof}

\subsubsection{Computing Resource Allocation of FVs}\label{FV_allocation}

\par Considering that the FVs may be unwilling to disclose private information to the MBS, we present a contract theory-based incentive mechanism to motivate FVs to share the idle resources. Specifically, the MBS offers serial contracts to different FVs, which specifies the computing resources to be shared and the corresponding utility that the FVs will receive. The FVs then decide whether to accept or decline the contract based on the obtained utility. This incentive mechanism is designed to stimulate FVs to truthfully reveal the available computing resources and collaborate in resource sharing.

\par \textbf{\textit{1) Utility Functions of FVs and MBS.}} We present utility functions of FVs and MBS to model the interaction between FVs and the MBS under incomplete information. 

\par \textbf{FV Type.} Considering that different FVs have varying computing resources, the concept of FV type is first introduced to quantify their resource sharing willingness. Intuitively, higher-type FVs are more inclined to contribute resources than the lower-type FVs. We define the type of the $l$th FV as follows:
\begin{equation}
    \theta_l = \sigma_l f_l^{\text{max}},
\end{equation}

\noindent where $\sigma_l$ indicates the strength of the willingness to contribute resources and $f_l^{\text{max}}$ denotes the maximum computational resource that an FV can contribute. Specifically, the MBS classifies the FVs into $L$ types, denoted as $\Theta=\{\theta_1,\theta_2,\ldots ,\theta_L\}$, which are sorted in ascending order such that $\theta_1<\theta_2<\cdots <\theta_L$. Then, the MBS sorts the FVs according to their willingness of resource sharing. Additionally, a contract item $(f_l(t),w_l(t))$ is designed for each type of FV, where $f_l(t)$ is the computing resources allocated by FV with type $\theta_l$, and $w_l(t)$ denotes the corresponding rewards that the FVs receive by allocating the resource.

\par \textbf{Utility Function of FVs.} The utility function of type $\theta_l$ FV that accepts the contract item $(f_l(t),w_l(t))$ can be calculated as the difference between the reward and costs, which is as follows:
\begin{equation}\label{utility_type_l}
\begin{aligned}
    U_l^{\text{FV}}(f_l(t),w_l(t))&=\theta_lw_l(t)-eE_l(t)\\
    &=\theta_lw_l(t)-e\kappa^{\text{FV}} C_n(t)f_l^2(t),
\end{aligned}
\end{equation}

\noindent where $e$ represents the unit cost of energy consumption, and $E_l(t)$ denotes the energy consumed by the FV for task computing. 

\par \textbf{Utility Function of the MBS.} The MBS does not know the exact types of the FVs due to the information asymmetry. Instead, the MBS knows the probability of the types derived from historical observations. We suppose that there are $L$ types of FVs known to the MBS, and each FV is independently classified as type $\theta_l$ with the same probability $\lambda_l$. According to the types of FVs, utility of the MBS is calculated as the payment received from the TVs minus the cost incurred in acquiring computing resources from the FVs~\cite{Diamanti2022Trading}, which is as follows:
\begin{equation}
    U^{\text{MBS}}(f_l(t),w_l(t),c)=\sum_{l=1}^{L}cM_lf_{l}(t)-\sum_{l=1}^{L}M_l\theta_lw_l(t),
\end{equation}

\noindent where $c$ represents the unit cost of the computing resources and $M_l=\lambda_lM$ denotes the total number of type $\theta_l$ FVs.

\textbf{\textit{2) Subproblem Formulation.}} We formulate the subproblem of FV resource computing. First, we present the feasibility conditions based on the utility functions. Specifically, the feasible contract $(f_l(t),w_l(t))$ should satisfy the individual rationality (IR) and incentive compatible (IC) conditions, which are given in Definitions \ref{Definition_IR} and \ref{Definition_IC}.

\begin{definition}
\label{Definition_IR}
The IR constraint indicates that a non-negative utility should be assigned to an FV if it agrees to the contract term. Therefore, the IR condition is formally expressed as: 
\begin{equation}
\label{IR}
    \theta_lw_l(t)-e\kappa^{\text{FV}} C_n(t)f_l(t)^2 \geq 0,\forall l \in\{1,2,\ldots,L\}.
\end{equation}
\end{definition}

\begin{definition}
\label{Definition_IC}
The IC constraint guarantees that an FV of type $\theta_l$ prefers the contract item $(f_{l}(t),w_l(t))$ over any other contract item $(f_{j}(t),w_j(t)),\forall j\in \{1,2,\ldots, L\}, j\neq l$. Therefore, the IC condition is given as:
\begin{equation}\label{IC}
    \theta_lw_l(t)-e\kappa^{\text{FV}} C_n(t)f_l^2(t)\geq \theta_lw_j(t)-e\kappa^{\text{FV}} C_n(t)f_j^2(t).
\end{equation}
\end{definition}

\par According to the contract theory, the subproblem of FV computing resource allocation can be reformulated by maximizing the utility of the MBS while considering the IR and IC constraints, which is as follows:
\begin{subequations}
    \begin{alignat}{2} 
    \mathbf{SP2}: \, &\max_{f_l(t),w_l(t)}\big(\sum_{l=1}^{L}cM_lf_{l}(t)-\sum_{l=1}^{L}M_l\theta_lw_l(t)\big)\label{sp2_a}\\
    \text{s.t.} \,
    &0\leq f_l(t)\leq f_l^{\text{max}},\forall l\in\{1,2,\ldots,L\},\label{p4_c_f}\\
    &(\ref{IR}),(\ref{IC}) \notag.
    \end{alignat}
    \label{P4}
\end{subequations}
\vspace{-0.8em}
\par The optimization problem $\mathbf{SP2}$ is difficult to be solved due to the complex constraints incurred by $L$ IR conditions and $L(L-1)$ IC conditions. Therefore, we will analyze the contract properties and reduce the constraints as follows. 

\textbf{\textit{(3) Properties of Feasible Contracts.}} Based on the definitions of the IR and IC constraints, several necessary conditions of feasible contracts can be derived.

\par\textbf{First}, we can derive from the IC constraint that a higher type of FV receives a higher reward. Conversely, a higher reward received by an FV indicates a higher type of the FV. Moreover, if two FVs have the same types, they receive the same rewards. This can be concluded in \textcolor{b}{Lemma \ref{lemma_1}.}

\begin{lemma}
\label{lemma_1}	
For any feasible contract item $(f_l(t), w_l(t))$, it holds that $w_i(t) > w_j(t)$ if and only if $\theta_i > \theta_j$, and $w_i(t) = w_j(t)$ if and only if $\theta_i = \theta_j$, $\forall i, j \in \{1, 2, \ldots, L\}$.
\end{lemma}

\begin{proof}
The proof is provided in Appendix C of the supplemental material.
\end{proof}

\par\textbf{Second}, by generalizing the results of Lemma \ref{lemma_1}, we can know that the rewards received by the FVs are monotonic with respect to the types, as given in Theorem \ref{theorem_1}.

\begin{theorem}
\label{theorem_1}
Monotonicity of rewards. For any feasible contract, the rewards for different types of FVs are as follows:
\begin{equation}
0< w_1(t)<\cdots<w_i(t)<\cdots<w_L(t).
\end{equation}
\end{theorem}

\begin{proof}
The proof is presented in Appendix D of the supplemental material.
\end{proof}

\par\textbf{Third}, according to the IC conditions, we know that the higher reward of an FV indicates that it allocates more computing resources. Conversely, the more computing resource allocation of the FV results in a higher reward. Moreover, if two FVs contribute the same amount of computing resources, they receive identical rewards. This is mathematically presented in Lemma \ref{lemma_2}.

\begin{lemma}
\label{lemma_2}	
For any feasible contract $(f_l(t), w_l(t))$, it holds that $w_i(t) > w_j(t)$ if and only if $f_i(t) > f_j(t)$. Additionally, $w_i(t) = w_j(t)$ if and only if $f_i(t) = f_j(t)$, $\forall i, j \in \{1, 2, \ldots, L\}$.
\end{lemma}	

\begin{proof}
The proof is presented in Appendix E of the supplemental material.
\end{proof}

\par\textbf{Finally}, based on Theorem \ref{theorem_1} and Lemma \ref{lemma_2}, we know that the computing resource allocation is monotonic with respect to the reward, as presented in Theorem \ref{theorem_2}.

\begin{theorem} 
\label{theorem_2}
Monotonicity of resource allocation. For any feasible contract $(f_l(t),w_l(t))$, the computing resource allocation for different types of FVs is as follows:
\begin{equation}
   0\leq f_1(t)<\cdots<f_i(t)<\cdots<f_L(t).
\end{equation}

\end{theorem} 

\begin{proof}
The proof is presented in Appendix F of the supplemental material.
\end{proof}

\textbf{\textit{(4) Constraints Reduction.}} Due to the large number of constraints presented above, problem $\mathbf{SP2}$ is complex to be solved. Therefore, we reduce some of the constraints as follows.

\par\textbf{First}, we reduce the IR constraints to a single IR constraint, as satisfying the IR constraint for the lowest type ensures that the IR constraints for higher types will also be satisfied, which is presented in Lemma \ref{lemma_3}.

\begin{lemma}
\label{lemma_3}
If the IR constraint of type $\theta_1$ FVs is satisfied, then the IR constraints of type $\theta_l$ ($l\in \{2,3, \ldots ,L\}$) FVs will also be satisfied. That is, the IR constraints is reduced as:
\begin{equation}\label{simplified_IR}
    \theta_1w_1(t)-e\kappa^{\text{FV}} C_n(t)f_1^2(t)\geq 0.
\end{equation}
\end{lemma}

\begin{proof}
The proof is presented in Appendix G of the supplemental material.
\end{proof}

\par\textbf{Second}, we reduce the IC constraints by introducing four notations, i.e., downward incentive constraints (DICs), upward incentive constraints (UICs), local DICs (LDICs), and local UICs (LUICs). The DIC is defined as the IC constraint between type $\theta_i$ FV and type $\theta_j$ FV ($j\in\{1,\dots,i-1\}$), and the LDIC is defined as the IC constraint between FV of type $\theta_i$ and FV of type $\theta_{i-1}$. Similarly, the IC constraints between FV of type $\theta_i$ and FV of type $\theta_j$ ($j\in\{i+1,\dots,L\}$) are referred as UIC, and the IC constraint between FV of type $\theta_{i-1}$ and FV of type $\theta_{i+1}$ is referred as LUIC~\cite{Kazmi2022ANovel}. Based on these definitions, we reduce the IC constraints to LDICs and the IR constraints to LUICs, as given in Lemma \ref{lemma_4}.

\begin{lemma}
\label{lemma_4}
 The IC constraints can be reduced to LDICs as: 
\begin{equation}
\label{LDICs}
\begin{aligned}
\theta_iw_i(t)-e\kappa^{\text{FV}} C_n(t)f_i^2(t)\geq \theta_iw_{i-1}(t)-e\kappa^{\text{FV}} C_n(t)f_{i-1}^2(t),\\ 
\forall i\in \{2,\ldots,L\},
\end{aligned}
\end{equation}

\noindent and to LUICs as: 
\begin{equation}
\begin{aligned}
\theta_iw_i(t)-e\kappa^{\text{FV}} C_n(t)f_i^2(t)\geq \theta_iw_{i+1}(t)-e\kappa^{\text{FV}} C_n(t)f_{i+1}^2(t),\\ 
\forall i \in\{1,\ldots,L-1\}.
\end{aligned}
\end{equation}
\end{lemma}

\begin{proof}
The proof is presented in Appendix H of the supplemental material.
\end{proof}

\par\textbf{Finally}, we further simplify the LDICs, LUICs and IR constraint. Specifically, to maximize the utility of the MBS, the LDICs and IR constraints for FV of type $\theta_1$ can be enforced as tight, which is presented in Lemma \ref{lemma_5}. Furthermore, the LUICs can be replaced by the LDICs. In other words, if the LDICs are satisfied, the LUICs can be satisfied, as given in Lemma \ref{lemma_6}. 

\begin{lemma}
\label{lemma_5}
If the utility of the MBS is maximized, then both the LDICs and the IR constraints for FVs of type $\theta_1$ must be tight.
\begin{subequations}
    \begin{alignat}{2}
    &\theta_lw_l(t)-e\kappa^{\text{FV}} C_n(t)f_l^2(t)=\theta_lw_{l-1}(t)-e\kappa^{\text{FV}} C_n(t)f_{l-1}^2(t),\nonumber \\
    &\hspace{5.4cm} \forall l\in\{2,\ldots,L\},\\
    &\theta_1w_1(t)-e\kappa^{\text{FV}} C_n(t)f_1^2(t)=0.
    \end{alignat}
\end{subequations}
\end{lemma}

\begin{proof}
The proof is presented in Appendix I of the supplemental material.
\end{proof}

\begin{lemma}
\label{lemma_6}
If all LDICs are satisfied, then all LUICs also hold.
\end{lemma}

\begin{proof}
The proof is presented in Appendix J of the supplemental material.
\end{proof}

\par \textbf{\textit{(5) Optimal FV Computing Resource Allocation}}. Upon reducing the IR and IC constraints, problem $\mathbf{SP2}$ can be simplified as:

\vspace{-0.8em}

{\small
\begin{subequations}\label{P5}
    \begin{alignat}{2} 
   \mathbf{SP2.1}: \,  &\max_{f_l(t),w_l(t)}(\sum_{l=1}^{L}cM_lf_{l}(t)-\sum_{l=1}^{L}M_l\theta_lw_l(t))\\
    \text{s.t.} \,
    &\theta_1w_1(t)-e\kappa^{\text{FV}} C_n(t)f_1^2(t)=0,\label{p5_c1}\\
    &\theta_lw_l(t)-e\kappa^{\text{FV}} C_n(t)f_l^2(t)=\theta_lw_{l-1}(t)-\nonumber \\
    &e\kappa^{\text{FV}} C_n(t)f_{l-1}^2(t),\forall l\in\{2,\ldots,L\},\label{p5_c2}\\
    &0\leq w1(t)\leq\cdots\leq w_L(t),\forall l\in\{1,\ldots,L\},\label{p5_w}\\
    &w_L(t)<c,  \label{p5_w_c}\\
    &(\ref{p4_c_f}) \notag.
    \end{alignat}
\end{subequations}
}

\begin{theorem}
\label{theorem_p5_convex}
Problem $\mathbf{SP2.1}$ is a convex optimization problem.
\end{theorem}

\begin{proof}
The proof is presented in Appendix K of the supplemental material.
\end{proof}

\par According to Theorem \ref{theorem_p5_convex}, problem $\mathbf{SP2.2}$ is convex and has a global optimal solution. Therefore, the existing optimization tools such as CVX can be applied to obtain the optimal computing resource allocation $f_l^*(t)$. The main processes of FV computing resource allocation are given in Algorithm \ref{resource_allocation_algorithm}. Specifically, problem $\mathbf{SP2.2}$ is solved by using the CVX method to obtain $w_l^*(t)$ (lines \ref{calculate_w_1} to \ref{calculate_w_2}). Then, the rewards of each FV of type $l$ is converted into the rewards for each FV $m$ (lines \ref{convert_w_1} and \ref{convert_w_2}). Subsequently, the computing resources allocation $f_l^*(t)$ for each FV of type $l$ is iteratively calculated based on constraint \eqref{p5_c2} (lines \ref{calculate_f_1} to \ref{calculate_f_2}). Finally, the optimal computing resources that each FV $m$ should allocate to TV $n$ is calculated (lines \ref{convert_f_1} and \ref{convert_f_2}).

\begin{algorithm}[]
    \label{resource_allocation_algorithm}	
    \SetAlgoLined
    \KwIn{$e,\theta, L, \kappa^{\text{FV}}, M_l, c, C_n(t)$}
    \KwOut{$f_{n,m}^*(t)$}
    \textbf{Initialization:} $f_{n,m}(t)=0$, $w_{n,m}(t)=0$;\\
    \For{$n\in\mathcal{N}$}{
        Set $w_l(t)=0$;\label{calculate_w_1}\\
        \textbf{cvx\_begin}\\
        \quad Defining variables $w(L)$;\\
        \quad Objective function (21) of the supplementary document;\\
        \quad Constraints \eqref{p5_w}, \eqref{p5_w_c}, \eqref{p4_c_f};\\
        \textbf{cvx\_end}\label{calculate_w_2}\\
        Mapping $w_{m}(t) \leftarrow w_l(t)$;\label{convert_w_1}\\
        $w_{n,m}(t) = w_m(t)$;\label{convert_w_2}\\
        Set $f_l(t)=0$;\label{calculate_f_1}\\
        $f_1(t)=\sqrt{\frac{\theta_1w_1(t)}{e\kappa^{\text{FV}} C_n(t)}}$;\\
        \For{$l=2$ to $L$}{
            Calculate $f_l(t)$ according to \eqref{p5_c2};\label{calculate_f_2}\\
        }
        Mapping $f_m(t) \leftarrow f_l(t)$;\label{convert_f_1}\\
        $f_{n,m}(t) = f_m(t)$;\label{convert_f_2}\\
    }
    \Return  $f_{n,m}^*(t)$.
  \caption{FV Computing Resource Allocation.}
\end{algorithm}

 \vspace{-1.5em}

\subsection{Task Offloading}

\par The task offloading decision is obtained by proposing a two-sided matching method. Considering the heterogeneous preferences between TVs and servers, we employ the matching game to construct a mutual-satisfied matching between TVs and servers for task offloading. Specifically, the matching between TVs and servers is defined as follows.

\begin{definition}
\label{def_match}
The current matching is defined as $(\mathcal{A},\mathcal{P}(t),\Pi(t))$, where  
\begin{itemize}
\item $\mathcal{A}=(\mathcal{N},\mathcal{M} \cup \mathcal{K})$ consists of two distinct sets of agents, i.e., TVs and servers. 
\item $\mathcal{P}(t)=(\mathcal{P}_{n}(t),\mathcal{P}_{s}(t))$ consists of the preference lists of TVs and servers in time slot $t$. Specifically, each TV $n\in\mathcal{N}$ maintains a descending ordered preference list on the servers, denoted as $\mathcal{P}_{n}(t)=\{s|s\in\mathcal{M} \cup \mathcal{K}, s\succ_{n}s^{\prime}\}$, where $\succ_{n}$ indicates the preference of TV $s\in\mathcal{M} \cup \mathcal{K}$ on servers. Similarly, each server $s$ has a descending ordered preference list on the TVs, denoted as $\mathcal{P}_{s}(t)=\{n|n\in\mathcal{N}, n\succ_{s}n^{\prime}\}$, where $\succ_{s}$ represents the preference of server $s$ on TVs.

\item $\Pi(t)\subseteq \mathcal{N}\times ( \mathcal{M}  \cup \mathcal{K})$ represents the matching between the TVs and servers. Each TV $n$ can be matched with at most one server, i.e., $\Pi_{n}(t)\in \mathcal{M} \cup \mathcal{K}$, while each server $s$ can be matched with multiple TVs, i.e., $\Pi_s(t)\in \mathcal{N}$.
\end{itemize}
\end{definition}

\par The two-sided matching method is shown in Algorithm \ref{algo_match}, which is elaborated as follows.

\subsubsection{Preference List Construction}

\par According to Definition \ref{def_match}, the preference lists for TVs and servers are constructed as follows.

\textit{\textbf{Preference List for TVs.}} For TVs, each TV $n$ aims to minimize the task completion delay by selecting a satisfied server $s$ for task offloading. Thus, the preference value of TV $n$ on server $s$ is given as:
\begin{equation}
\label{prefer_task}
    \Phi_{n,s}(t)=1/T_n(t).
\end{equation}

\noindent Then, the preference list of each TV is constructed by sorting the servers based on the preference values in descending order, which is given as:
\begin{equation}
    s\succ_{n}s^{\prime} \iff \Phi_{n,s}(t)\geq \Phi_{n,s^{\prime}}(t).
\end{equation}

\textit{\textbf{Preference List for Servers.}} For servers, each servers prefer to provide computing service for TVs with low energy consumption. Thus, the preference value of server $s$ on TV $n$ is given as:
\begin{equation}
\label{prefer_fog}
    \Phi_{s,n}(t)=1/E_{n,s}(t).
\end{equation}

\noindent Similarly, the preference list of each server is constructed by sorting the TVs based on the preference values in descending order, which is given as:
\begin{equation}
    n\succ_s n^{\prime} \iff \Phi_{s,n}(t)\geq\Phi_{s,n^{\prime}}(t).
\end{equation}

\subsubsection{Matching Construction}

\par According to the preference lists, the two-sided matching between TVs and servers is constructed according to the following steps.

\par First, each TV $n$ selects the most preferred server $s$ and temporarily adds it to the matching list as follows:
\begin{equation}
\label{update_ml_n}
    \Pi_n(t) = \Pi_n(t) \cup s.
\end{equation}

\par Then, if server $s$ is the most preferred server of TV $n$, TV $n$ is temporarily added to the matching list of server $s$, which is as follows:
\begin{equation}
    \Pi_s(t) = \Pi_s(t) \cup n.\label{update_ml_s}
\end{equation}

\par Moreover, each server $s$ updates its matching list by removing the lower-priority TVs, ensuring that the allocated computational resources do not exceed its maximum resources, which is as follows:
\begin{equation}
\label{retain_TV}
 \sum_{s\in\Pi_s(t)}f_{n,s}^*(t)\leq f_s^{\text{max}}, \  \Pi_s(t) = \Pi_s(t)\backslash D_s,
\end{equation}

\noindent where $D_s$ represents the set of lower-priority TVs.

\par Additionally, the lower-priority TVs are added to the rejection set $D_s$, which is as follows:
\begin{equation}
    \mathcal{R} = \mathcal{R} \cup D_s,\label{update_rejection_set}
\end{equation}

\noindent where $\mathcal{R}$ denotes the TV rejection set.

\par Finally, the matching list $\Pi_{n^{\prime}}(t)$ and preference list $\Phi_{n^{\prime},s}(t)$ are updated for each TV $n^{\prime}\in D_s$ that has been rejected by server $s$, which is as follows: 
\begin{equation}
\label{update_matchlist_preferlist}
 \Pi_{n^{\prime}}(t) =  \varnothing, \  \Phi_{n^{\prime},s}(t) = \Phi_{n^{\prime},s}(t) \backslash s.
\end{equation}

\par The above steps are repeated until all TVs are paired with a server, or until any unmatched TVs have been rejected by all servers.

\subsubsection{Matching Result Analysis}

\par According to Definitions \ref{def_block} and \ref{def_stable}, the result of task offloading obtained by the two-sided matching method is stable and weak-Pareto optimal, as given in Theorems \ref{the_stable} and \ref{the_pareto}. \textcolor{b}{Moreover, the two-sided matching will terminate within a finite number of iterations, as presented in Theorem \ref{the_finite}.}

\begin{definition}
\label{def_block}
    Blocking pair. Assuming that a TV $n$ and a server $s$ are not matched in the current matching result $\Pi(t)$, the matching $\Pi(t)$ is considered blocked by the blocking pair $(n,s)$ if TV $n$ and a server $s$ prefer each other over their current pairs.
\end{definition}

\begin{definition}
\label{def_stable}
    Stable matching. The matching $\Pi(t)$ is stable if and only if there are no blocking pairs. 
\end{definition}

\begin{theorem}
\label{the_stable}
The matching $\Pi$ proposed by this work is stable for each TV $n$ and server $s$ ($n\in\mathcal{N}$ and $s\in\mathcal{M}\cup\mathcal{K}$).
\end{theorem}
\begin{proof}
The proof is presented in Appendix L of the supplemental material.
\end{proof}

\begin{theorem}
\label{the_pareto}
    The matching $\Pi$ is weak-Pareto optimal for each TV $n$ and server $s$ ($n\in\mathcal{N}$ and $s\in\mathcal{M}\cup\mathcal{K}$).
\end{theorem}
\begin{proof}
The proof is presented in Appendix M of the supplemental material.
\end{proof}
{\color{b}
\begin{theorem}
\label{the_finite}
  The two-sided matching will terminate within a finite number of iterations.
\end{theorem}
\begin{proof}
 The proof is presented in Appendix N of the supplemental material.
\end{proof}
}

\par The main processes for task offloading are presented in Algorithm \ref{algo_match}. First, problem $\mathbf{SP1}$ is solved to obtain the computing resource allocation for RSUs (line \ref{f_allocation}). Then, the preference list for TVs and servers are established (lines \ref{calculate_TV_pl} and \ref{sort_pl}). Subsequently, each TV $n$ selects the most preferred server $s$, temporarily adds server $s$ to the matching list of TV $n$, and adds TV $n$ to the matching list of server $s$ (lines \ref{select_p_s} to \ref{add_n_to_s}). Additionally, each server $s$ retains the qualifying TVs to update its matching list (line \ref{update_server_list}). Moreover, add the unqualified TVs to the overall rejection set (line \ref{update_R}). Finally, update the matching list and preference list for each TV $n$ that has been rejected by server $s$ (line \ref{update_ml_pl}).

\vspace{-1em}
\begin{algorithm}[]	
    \label{algo_match}
    \SetAlgoLined
    \KwIn{The sets of TVs $\mathcal{N}$ and servers $\mathcal{M}\cup\mathcal{K}$}
    \KwOut{The optimal matching list $\Pi(t)^*$ and the offloading strategy $\mathbf{O}^*(t)$}
    \textbf{Initialization:}
    $\Pi_n(t) = \varnothing$, $\Pi_s(t) = \varnothing$, $\mathcal{R}=\mathcal{N}$;\\
    \For{$n \in \mathcal{N}$}{
        \For{$s \in \mathcal{M}\cup\mathcal{K}$}{
            Obtain $f_{n,s}^*$ through resource allocation;\label{f_allocation}\\
            Obtain preference values as Eqs. \eqref{prefer_task} and \eqref{prefer_fog};\label{calculate_TV_pl}
        }
    }
    Sort the preference lists of TVs and servers in descending order of preference values.\label{sort_pl}\\
    \While{$\mathcal{R}$ is not empty}{
        \For{$n \in \mathcal{R}$}{
            Select the most preferred server $s$;\label{select_p_s}\\
            Establish the matching lists as Eqs.~\eqref{update_ml_n} and ~\eqref{update_ml_s};\label{add_n_to_s}\\
        }
        \For{$s \in \mathcal{M}\cup\mathcal{K}$}{
            \uIf{$\Pi_s(t)$ is not empty}{
                Update the matching lists as Eq.~\eqref{retain_TV};\label{update_server_list}\\
                Update the rejection set as Eq.~\eqref{update_rejection_set};\label{update_R}\\
                \For{$n^{\prime} \in D_s$}{
                    Update $\Pi_{n^{\prime}}(t)$ and $\Phi_{n^{\prime},s}(t)$ as Eq.~\eqref{update_matchlist_preferlist};\label{update_ml_pl}\\
                }
            }
        }
    }
    \Return  $\Pi(t)$, $\mathbf{O}^*(t)=\{o_{n,s}(t)|s=\Pi_n(t)\}$.
  \caption{Task Offloading.}
\end{algorithm}


\subsection{Main Steps of JCRATOA and Performance Analysis}

\par In this section, we show the main steps and performance analyses of the proposed JCRATOA.

\subsubsection{Main Steps of JCRATOA}

\par The main steps of JCRATOA are presented in Algorithm \ref{JCRATOA_algorithm}. Specifically, the system time and delay are initialized (line 1). Then, in each time slot, the optimal computing resource allocation decisions of RSUs and FVs are obtained (lines \ref{JCRATOA_line_0} to \ref{JCRATOA_line_2}). Moreover, the decisions of task offloading are obtained (line \ref{JCRATOA_line_3}). Furthermore, the task completion delay for TV is calculated based on the decisions of computing resource allocation and task offloading (line \ref{JCRATOA_line_4}). Additionally, update the delay, available computing resources, and mobility states of vehicles (lines \ref{JCRATOA_line_5} to \ref{JCRATOA_line_6}).


\begin{algorithm}[]
    \label{JCRATOA_algorithm}	
    \SetAlgoLined
    \KwIn{$\mathcal{N}, \mathcal{M}, \mathcal{K}, \mathcal{T}$}
    \KwOut{System delay $SL$}
    \textbf{Initialization:} $t=0$, $SL=0$;\\
    \While{$t\leq T$}{\label{JCRATOA_line_0}
        Use Matlab fmincon tool to obtain \textcolor{b}{$f^*_{n,k}(t)$};\label{JCRATOA_line_1}\\
        Call Algorithm \ref{resource_allocation_algorithm} to obtain $f^*_{n,m}(t)$;\label{JCRATOA_line_2}\\
        Call Algorithm \ref{algo_match} to obtain $\Pi^*(t)$ and $\mathbf{O}^*(t)$;\label{JCRATOA_line_3}\\
        \For{$n\in\mathcal{N}$}{
            Calculate the delay $T_n(t)$;\label{JCRATOA_line_4}\\
            $SL(t)=SL(t)+T_n(t)$;\label{JCRATOA_line_5}\\
            Update the computing resources of TVs, FVs and RSUs;\label{JCRATOA_line_6}\\
        }
        Update the mobility of TVs and FVs;\label{JCRATOA_line_7}\\
        Update time $t=t+\tau$;\\
    }
    \Return $SL$.
  \caption{JCRATOA}
\end{algorithm}

 \vspace{-1.5em}

\subsubsection{\textcolor{b}{Complexity Analysis}}

\par For RSU computing resource allocation, the worst-case complexity of problem $\mathbf{SP1}$ is $\mathcal{O}(KN^2)$, where $K$ and $N$ denote the number of FVs and TVs, respectively~\cite{Jee2020Precoding}. For FV computing resource allocation, the worst-case complexity of problem $\mathbf{SP2.2}$ is $\mathcal{O}(NL^3)$, where $L$ is the number of FV types. For task offloading, we can derive that the complexity of preference list construction is $\mathcal{O}(N(M+K))$, where $M+K$ denotes the number of servers~\cite{Li2024Joint}. Moreover, for matching construction, in the worst case, any TV could be rejected $M+K$ times~\cite{Sun2025online}. Each rejection requires updating the preference list of at most $\min\{M+K, N\}$ servers in the next iteration. Therefore, the worst-case complexity of matching construction is $\mathcal{O}((M+K)(N+\min\{M+K,N\}))$, and the complexity of Algorithm \ref{algo_match} is $\mathcal{O}((M+K)(2N+\min\{M+K,N\}))$. In summary, the worst-case complexity of the proposed JCRATOA is $\mathcal{O}((M+K)(2N+\min\{M+K,N\})+KN^2+NL^3)$.

\begin{remark}
\par \textcolor{b}{Note that the complexity reduction for solving the DMOP can cause the performance degradation due to the smaller solution space resulting from problem decomposition. However, the proposed JCRATOA satisfies the requirements of vehicles while meeting the constraints of the system. This is because although decomposing the problem into subproblems reduces the solution space of each subproblem, the optimization objective and constraints of each subproblem remain consistent with those of the original problem, thereby ensuring the feasibility of the solution in meeting the constraints of the original problem.}
\end{remark}

\section{Simulation Results and Analysis}
\label{SIMULATION_RESULTS_AND_ANALYSIS}

\subsection{Simulation Setup}

\subsubsection{Parameters}

\par We consider a 3 km road, where 3 RSUs are deployed and 12 FVs are randomly located initially. Moreover, the communication coverage radius of each TV is 200 m. Additionally, the system timeline is set as 40 s, with the time slot of 1 s. The other parameters are listed in Table \ref{tab_simuParameter}.

\begin{table}[t]
	\setlength{\abovecaptionskip}{-5pt}%
	\setlength{\belowcaptionskip}{0pt}%
	\caption{Simulation parameters}
	\label{tab_simuParameter}
	\renewcommand*{\arraystretch}{1}
	\begin{center}
		\begin{tabular}{p{.06\textwidth}|p{.21\textwidth}|p{.13\textwidth}}
			\hline
			\hline
			\textbf{Symbol}&\textbf{Meaning}&\textbf{Default value}\\
			\hline
				$\mathcal{N}$&The number of TVs& [5, 30]
                \textcolor{b1}{\cite{Shah2022SDN}}\\
            \hline
			$D^{in}_n$&Task size& [300, 1000] KB
            \textcolor{b1}{\cite{Sun2024Bargain}}\\
            \hline
			$t_n^{max}$&Task deadline& $[0.5, 5]$ s
            \textcolor{b1}{\cite{Sun2024Bargain}}\\
            \hline
			$p$&Transmit power& $[20, 50]$ dBm\\
            \hline
			$N_0$&Noise power& $-98$ dBm
            \textcolor{b1}{\cite{Sun2024TJCCT}}\\   
            \hline
			$B$&Channel bandwidth& $[20, 40]$ MHz
            \textcolor{b1}{\cite{Chen2025Efficient}}\\  
            \hline
			$f_{n}$&Computing resources of TV $n$& $[0.5, 1]$ GHz
             \textcolor{b1}{\cite{Sun2024Bargain}}\\ 
            \hline
			$f_{m}$&Computing resources of FV $m$& $[1, 10]$ GHz
            \textcolor{b1}{\cite{Fan2023Joint}}\\  
            \hline
			$f_{k}$&Computing resources of RSU $k$& $30$ GHz
            \textcolor{b1}{\cite{Fan2023Joint}}\\    
			\hline
            \textcolor{b}{$\bar{\textbf{v}}_{v}$} & \textcolor{b}{Asymptotic mean of velocity}& \textcolor{b}{25 m/s}\\
            \hline
            \textcolor{b}{$\varsigma$}& \textcolor{b}{Asymptotic standard deviation of velocity}& \textcolor{b}{5}\\
            \hline
            \textcolor{b}{$\alpha$} & \textcolor{b}{Memory level of velocity}& \textcolor{b}{0.9}
            \textcolor{b1}{\cite{Sun2024TJCCT}}\\
            \hline
            \textcolor{b}{$H_k$} & \textcolor{b}{Effective antenna height at RSU $k$}& \textcolor{b}{10 m} \textcolor{b1}{\cite{3GPPTR389012020}} \\
            \hline
            \textcolor{b}{$H_n$} & \textcolor{b}{Effective antenna height at vehicle $n$}& \textcolor{b}{1.5 m} \textcolor{b1}{\cite{3GPPTR389012020}} \\
            \hline
            \textcolor{b}{$f_c$}&\textcolor{b}{Carrier frequency} &\textcolor{b}{5.9 GHz}\textcolor{b1}{\cite{ETSITR1032571}}\\
            \hline
            \textcolor{b}{$\vartheta^{\text{L}}$} & \textcolor{b}{Shadow fading}& \textcolor{b}{4 dB}\textcolor{b1}{\cite{3GPPTR389012020}} \\
            \hline
		\end{tabular}
	\end{center}
\end{table} 

\vspace{-1.5 em}

\subsubsection{Evaluation Metrics}

\par We evaluate the performance of the JCRATOA by presenting the following indicators. i) Average task completion delay $\frac{1}{N}\sum_{n\in \mathcal{N}}T_n(t)$, which indicates the average delay for completing a task. ii) Average task completion ratio ${N^{\text{succ}}(t)}/\sum_{t\in\mathcal{T}}\sum_{n\in\mathcal{N}}$, which indicates the average ratio of tasks that are completed, where $N^{\text{succ}}(t)$ represents the number of tasks that have been successfully completed. iii) Average energy consumption $\frac{1}{T}\sum_{t\in \mathcal{T}}\sum_{n\in \mathcal{N}}(E_n(t)+E_{n,s}(t))$, which indicates the average energy consumption during the system timeline. {\color{b1} iv) System throughput $\sum_{t \in T} \sum_{n \in N} D_n^{\text{succ}}(t)/T$, which indicates the amount of tasks successfully completed per unit time, where $D_n^{\text{succ}}(t)$ denotes the amount of tasks that is successfully completed by TV $n$ at time slot $t$, and $T$ denotes the system timeline. v) Resource utilization fairness $(\sum_{n=1}^N x_n)^2/(N \sum_{n=1}^N x_n^2)$, which indicates the fairness of computing resource allocation among vehicles by using the Jain's fairness index\cite{Sediq2013Optimal}, where $x_n=\sum_{t \in T}(o_{n, n}(t) f_n+\sum_{s \in K \cup M} o_{n, s}(t) f_{n, s}(t))$ denotes the total amount of computing resources allocated to TV $n$ during the system timeline.}


\subsubsection{Comparison Approaches}

\par We compare JCRATOA with the following baselines:

\begin{itemize}
\item All local offloading (ALO): All TVs process their tasks locally.

\item Nearest RSU offloading (NRO): The tasks of each TV are offloaded to the nearest RSU to which it is connected, and the computing resource allocation of the RSU is determined based on JCRATOA. 

\item Nearest FV offloading (NFO): The tasks of each TV are offloaded to its nearest FV within range, and the computing resource allocation of the FV is determined based on JCRATOA. 

\item Nearest server offloading (NSO):  An optimal server is selected for task offloading from the nearest RSU or FV based on which offers better performance. Additionally, the computing resource allocation of RSUs and FVs is determined based on JCRATOA. 

\item Kuhn–Munkres matching-based task offloading (KMMTO) \cite{Wei2023TBOMC}: The task offloading decision is made by using the Kuhn–Munkres matching method, and the computing resource allocation decision is determined based on JCRATOA. 

\item Binary reverse offloading and Lagrangian dual-based resource allocation (BROLDRA) \cite{Feng2022Latency}: The task offloading decision is made by using a reverse offloading method. Specifically, the tasks of a TV is first uploaded to the nearest RSU, which then decides to process the tasks directly or offload them to the FVs by employing the greedy searching. Moreover, the computing resource allocation is decided by employing the Lagrangian dual method.

\item \textcolor{b}{PPO-based task offloading and computing resource allocation (PTOCRA)\cite{Shang2024Joint}: The task offloading and computing resource allocation are decided by the PPO algorithm.}

\item \textcolor{b}{DDPG-based task offloading and computing resource allocation (DTOCRA) \cite{He2025Low}: The task offloading and computing resource allocation are decided by the DDPG algorithm.}

\end{itemize}

\subsection{System Performance}
In this section, we first evaluate the system performance of the JCRATOA over time with default parameters. Then, we examine the impacts of various parameters on the performance of the proposed JCRATOA.

\subsubsection{Performance Evaluation}

\par \textcolor{b1}{Figs. \ref{fig_time}(a), \ref{fig_time}(b), \ref{fig_time}(c), \ref{fig_time}(d), and \ref{fig_time}(e) illustrate the average task completion delay, average task completion ratio, system throughput, average energy consumption, and resource utilization fairness, respectively over time.} As shown in Fig. \ref{fig_time}, the proposed JCRATOA outperforms the benchmarks in terms of task completion delay, task completion ratio, system throughput, and resource utilization fairness, while exhibiting a relatively higher average energy consumption. For the benchmarks, several factors contribute to their inferior performances in terms of the task completion delay, task completion ratio, system throughput, and resource utilization fairness. {\color{b1}First, the approaches such as ALO, NRO, NFO, and NSO are based on nearest offloading strategies, which result in inefficient task processing and unfair resource utilization as traffic tends to be biased toward geographically close RSUs or FVs.} Additionally, the BROLDRA approach suffers from additional delays and task failures due to the task uploading and forwarding process, as well as the potential inefficient task offloading decisions made by the greedy search algorithm. {\color{b1}The resource allocation decision of BROLDRA also lacks efficiency in motivating FVs to participate in resource sharing, thereby resulting in poor fairness of resource utilization.} Furthermore, the inferior performance of KMMTO in task completion delay, task completion ratio, system throughput, and resource utilization fairness is mainly due to the relative high computational complexity of the Kuhn–Munkres matching method, which has a computational complexity of $\mathcal{O}^3(n)$. \textcolor{b1}{Finally, PTOCRA and DTOCRA show inferior performance in task completion delay, task completion ratio, system throughput, and resource utilization fairness, which is due to the long training periods caused by the complex and hybrid action spaces of the DMOP. Although they achieve advantages in energy consumption, this comes at the cost of task processing efficiency and resource utilization fairness, thus making them unsuitable for delay-sensitive VFC systems.}


\par \textcolor{b1}{For the proposed JCRATOA, its superior performance in task completion delay, task completion ratio, system throughput, and resource utilization fairness can be mainly attributed to two key factors. On the one hand, the contract theory-based inventive mechanism of resource allocation stimulates FVs to contribute their idle computational resources voluntarily, thus effectively expanding the system computing capacity and improving both throughput and resource utilization fairness. On the other hand, the two-sided matching method of task offloading can ensure that each task can be adaptively assigned to a suitable server, which shortens the task completion delay while maintaining a high completion ratio.} \textcolor{b1}{However, the relatively higher energy consumption stems from its prioritization of task processing over energy savings. This trade-off is essential in delay-sensitive VFC systems, where timely execution is critical to support computation-intensive and delay-critical tasks, particularly for safety-related applications. Moreover, vehicles and RSUs generally possess more sufficient and stable power supply, and short-term energy variations have limited immediate impact on real-time performance as their effects manifest over longer timescales. Despite the higher energy consumption, the proposed JCRATOA still meets the energy constraints of TVs and RSUs, as these constraints are considered in our DMOP.} \textcolor{b1}{Consequently, this set of simulation results indicates that the proposed JCRATOA is able to achieve superior performances in task completion delay, task completion ratio, system throughput, and resource utilization fairness, while effectively meeting the satisfying constraints.}

\begin{figure*}[!hbt] 
	\centering
	\setlength{\abovecaptionskip}{2pt}%
	\setlength{\belowcaptionskip}{2pt}%
	\subfigure[Average task completion delay]
	{
		\begin{minipage}[t]{0.18\linewidth}
			\centering
			\includegraphics[scale=0.28]{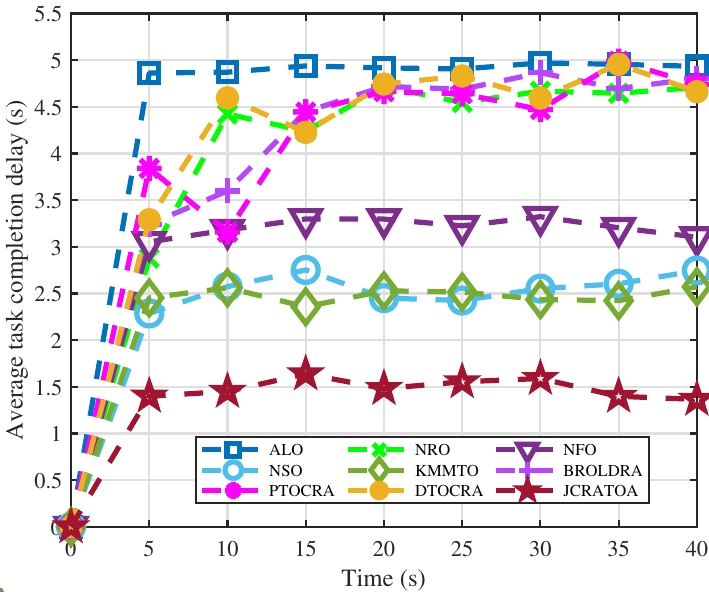}
		\end{minipage}
	}
	\subfigure[Average task completion ratio]
	{
		\begin{minipage}[t]{0.18\linewidth}
			\centering
			\includegraphics[scale=0.28]{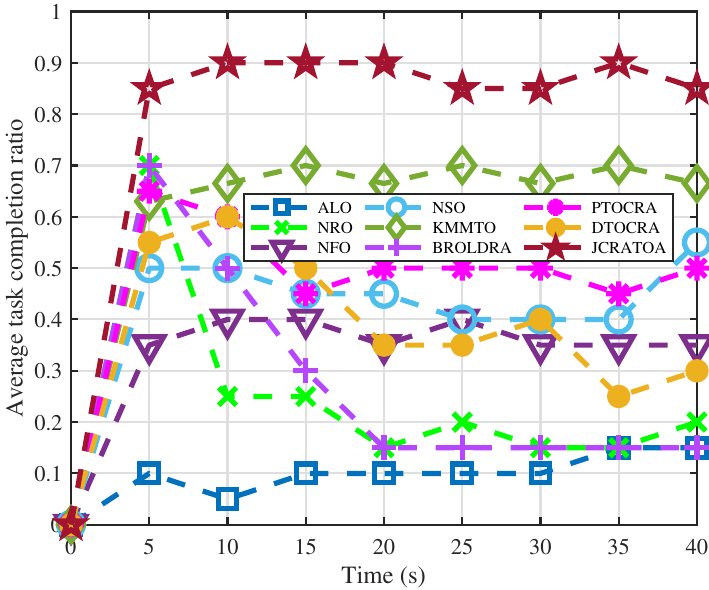}	
		\end{minipage}
	}
    \subfigure[System throughput]
	{
		\begin{minipage}[t]{0.18\linewidth}
			\centering
			\includegraphics[scale=0.28]{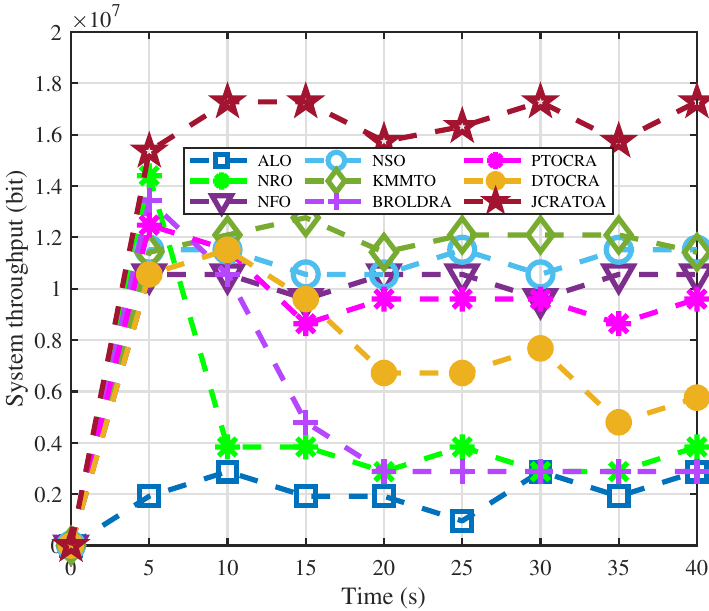}
		\end{minipage}
	}
    \subfigure[Resource utilization fairness]
	{
		\begin{minipage}[t]{0.18\linewidth}
			\centering
			\includegraphics[scale=0.28]{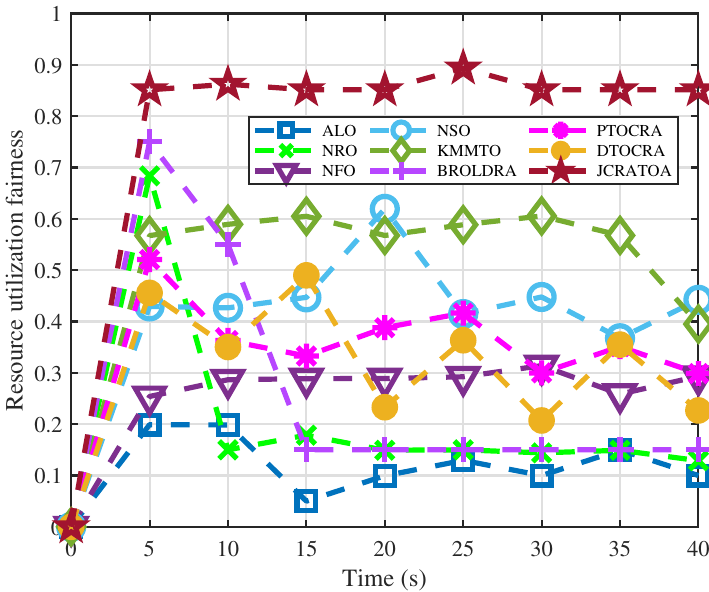}
		\end{minipage}
	}
    \subfigure[Average energy
consumption]
	{
		\begin{minipage}[t]{0.18\linewidth}
			\centering
			\includegraphics[scale=0.28]{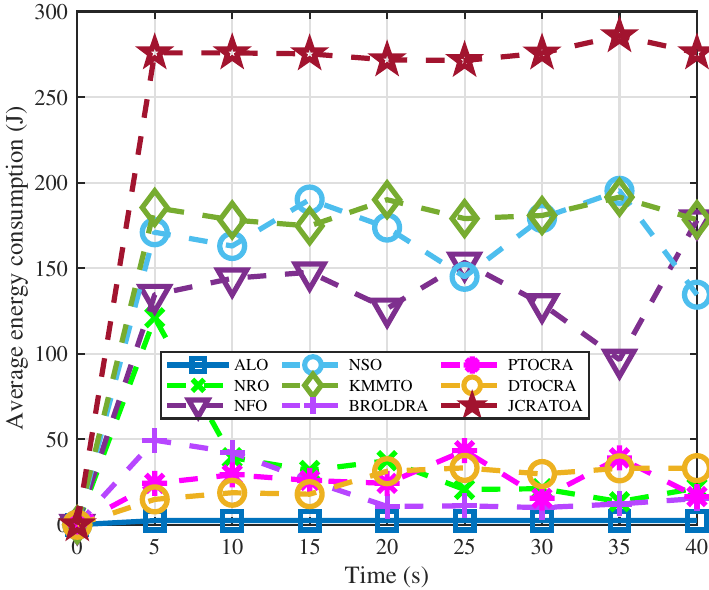}
		\end{minipage}
	}
	\centering
	\caption{\textcolor{b1}{System performance with respect to time.}}
	\label{fig_time}
	\vspace{-1.2 em}
\end{figure*}

\begin{figure*}[!hbt] 
	\centering
	\setlength{\abovecaptionskip}{2pt}%
	\setlength{\belowcaptionskip}{2pt}%
	\subfigure[Average task completion delay]
	{
		\begin{minipage}[t]{0.18\linewidth}
			\centering
			\includegraphics[scale=0.28]{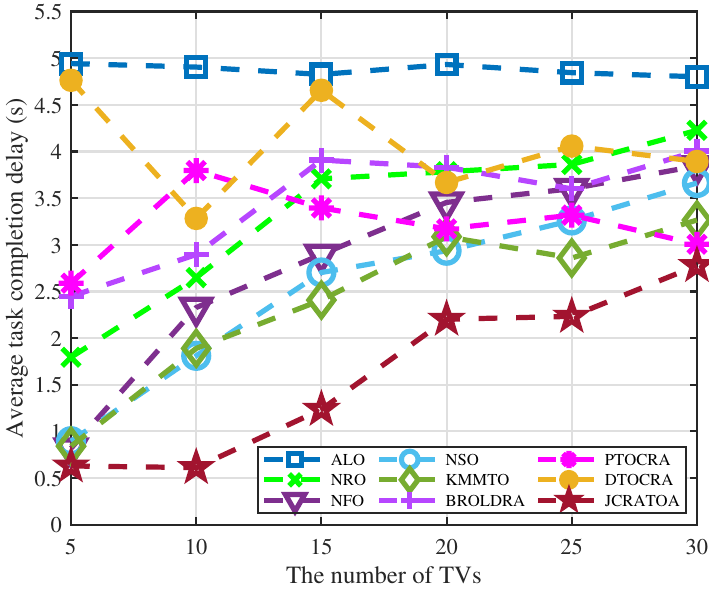}
		\end{minipage}
	}
	\subfigure[Average task completion ratio]
	{
		\begin{minipage}[t]{0.18\linewidth}
			\centering
			\includegraphics[scale=0.28]{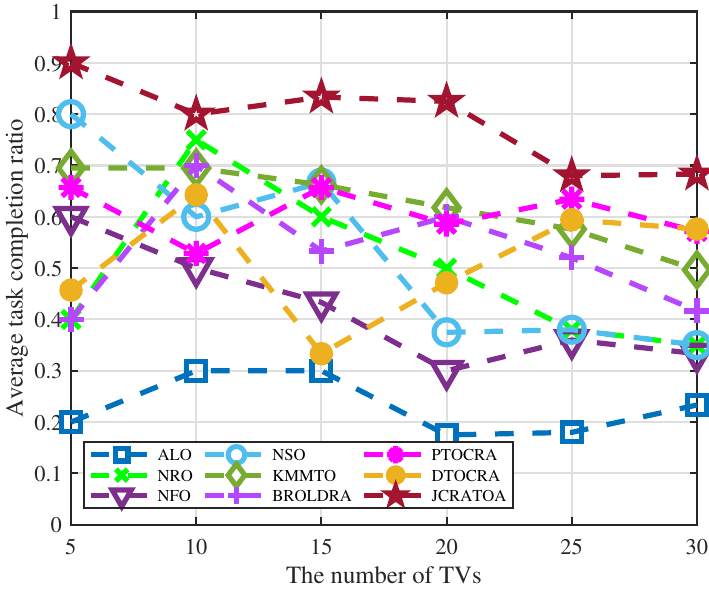}	
		\end{minipage}
	}
    \subfigure[System throughput]
	{
		\begin{minipage}[t]{0.18\linewidth}
			\centering
			\includegraphics[scale=0.28]{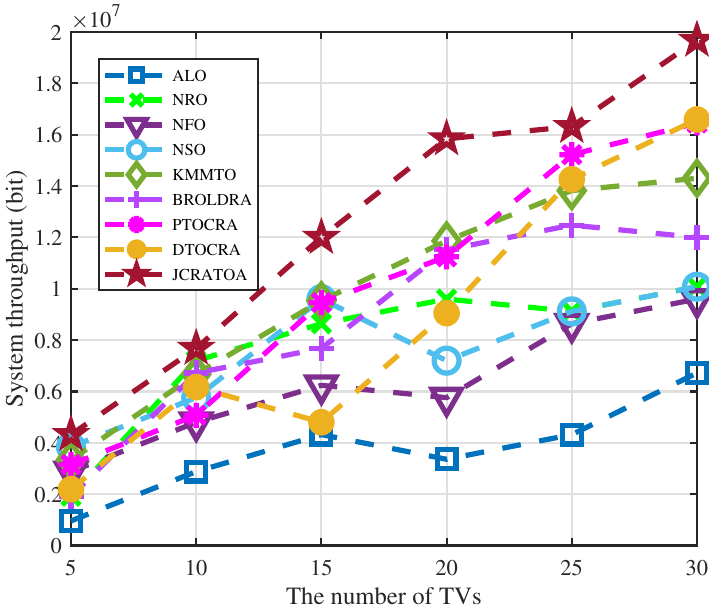}
		\end{minipage}
	}
    \subfigure[Resource utilization fairness]
	{
		\begin{minipage}[t]{0.18\linewidth}
			\centering
			\includegraphics[scale=0.28]{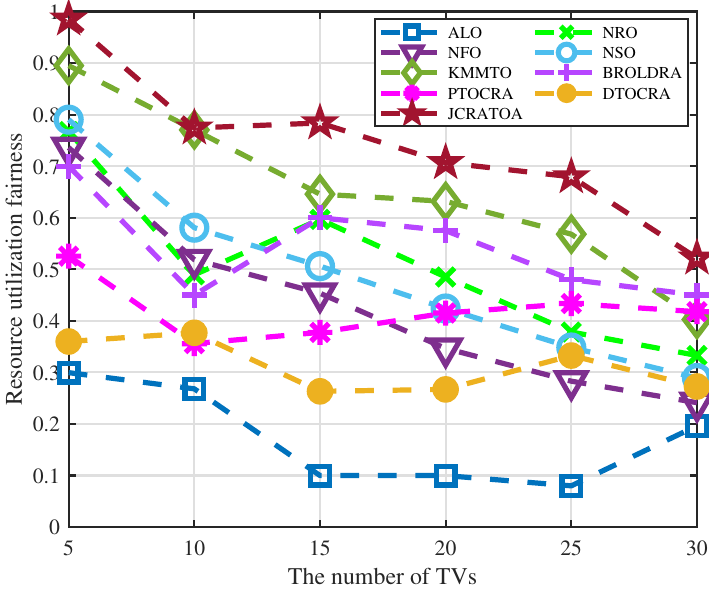}
		\end{minipage}
	}
    \subfigure[Average energy consumption]
	{
		\begin{minipage}[t]{0.18\linewidth}
			\centering
			\includegraphics[scale=0.28]{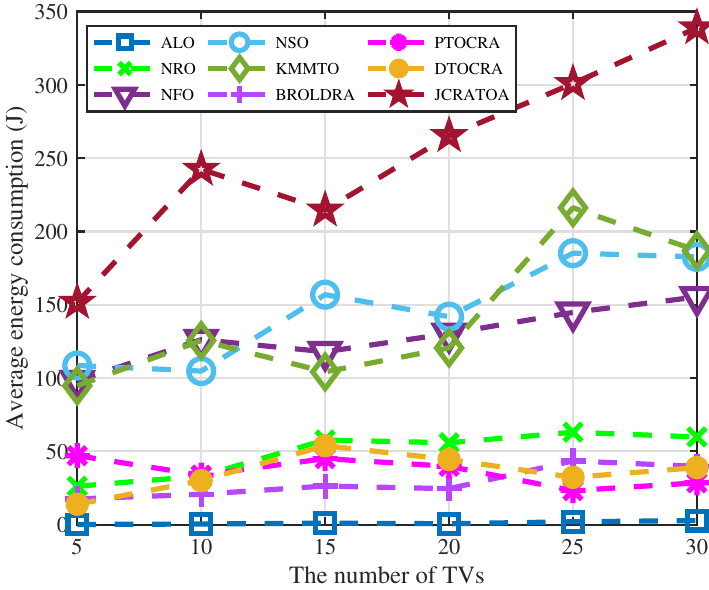}
		\end{minipage}
	}
	\centering
	\caption{\textcolor{b1}{System performance with different number of TVs.}}
	\label{fig_tasknodes}
	\vspace{-1.2em}
\end{figure*}

\begin{figure*}[] 
	\centering
	\setlength{\abovecaptionskip}{2pt}%
	\setlength{\belowcaptionskip}{2pt}%
	\subfigure[Average task completion delay]
	{
		\begin{minipage}[t]{0.18\linewidth}
			\centering
			\includegraphics[scale=0.28]{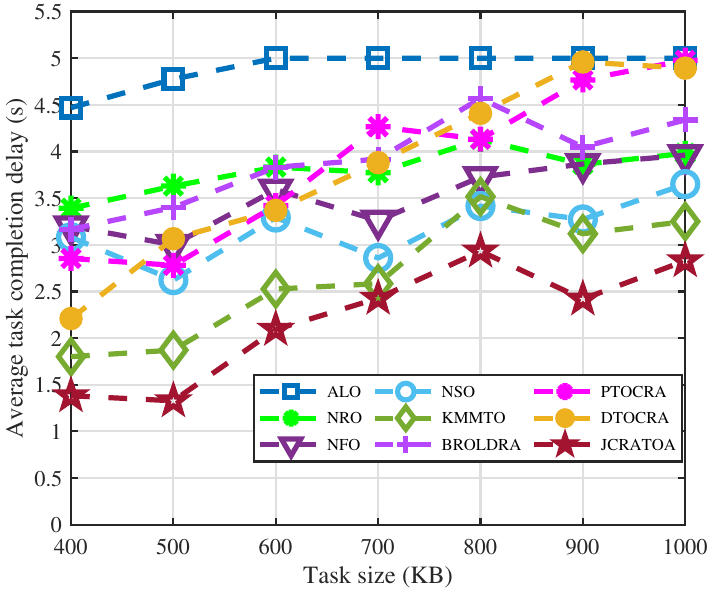}
		\end{minipage}
	}
	\subfigure[Average task completion ratio]
	{
		\begin{minipage}[t]{0.18\linewidth}
			\centering
			\includegraphics[scale=0.28]{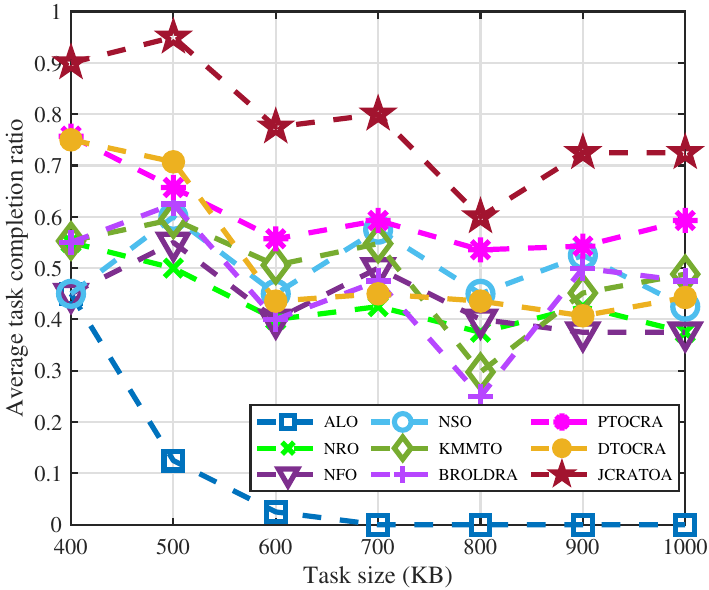}	
		\end{minipage}
	}
    	\subfigure[System throughput]
	{
		\begin{minipage}[t]{0.18\linewidth}
			\centering
			\includegraphics[scale=0.28]{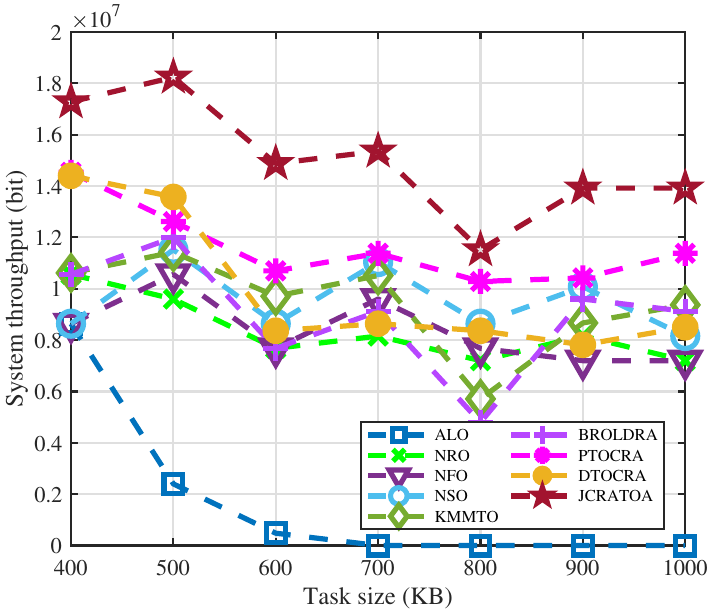}
		\end{minipage}
	}
    	\subfigure[Resource utilization fairness]
	{
		\begin{minipage}[t]{0.18\linewidth}
			\centering
			\includegraphics[scale=0.28]{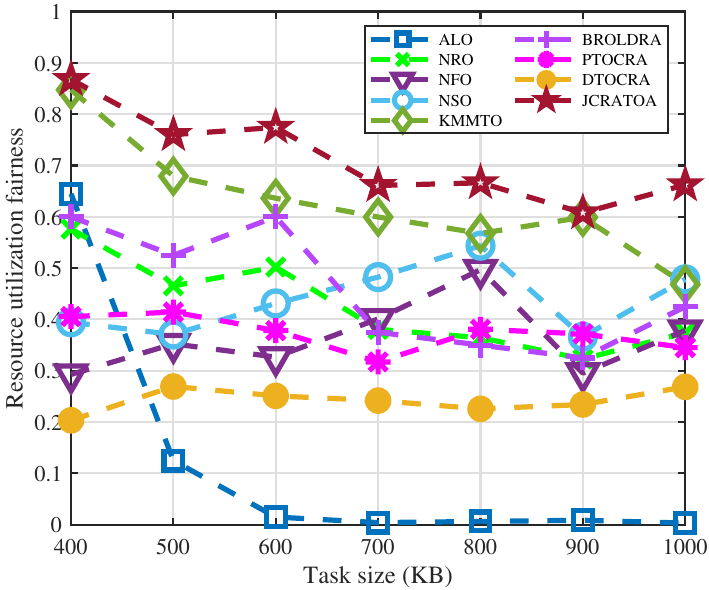}
		\end{minipage}
	}
    \subfigure[Average energy consumption]
	{
		\begin{minipage}[t]{0.18\linewidth}
			\centering
			\includegraphics[scale=0.28]{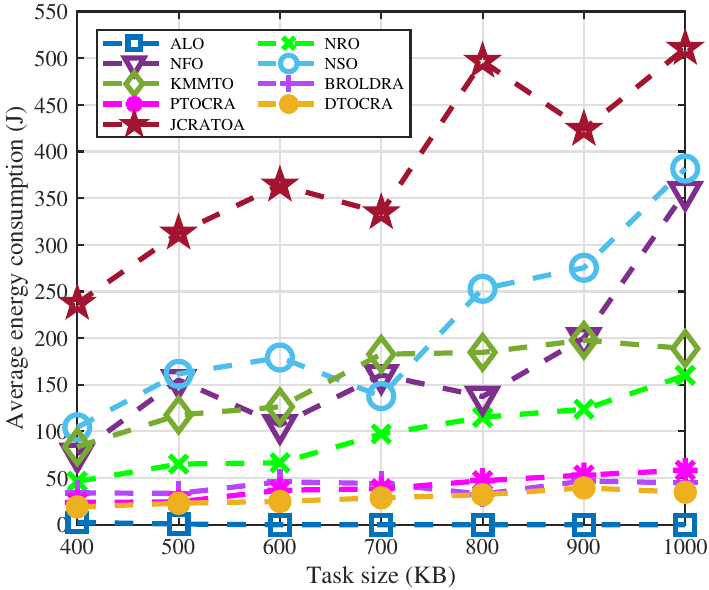}
		\end{minipage}
	}
	\centering
	\caption{\textcolor{b1}{System performance with different task sizes.}}
	\label{fig_tasksize}
	\vspace{-1.2em}
\end{figure*}

\subsubsection{Effect of TV Numbers}

\par Figs. {\color{b1}\ref{fig_tasknodes}(a), \ref{fig_tasknodes}(b), \ref{fig_tasknodes}(c), \ref{fig_tasknodes}(d), and \ref{fig_tasknodes}(e) show the impact of the number of TVs on average task completion delay, average task completion ratio, system throughput, resource utilization fairness, and average energy consumption,} respectively, for different approaches. \textcolor{b1}{As seen in Figs. \ref{fig_tasknodes}(a) to \ref{fig_tasknodes}(d), when the number of TVs increases, there is a general upward trend in task completion delay and system throughput, and a downward trend in task completion ratio and resource utilization fairness across all approaches. This is because more TVs cause heavier workloads and stronger resource contention, leading to longer processing delays and lower completion ratios. Meanwhile, the system throughput still increases as the computing capacity of RSUs and FVs is more fully utilized. However, the intensified load competition also widens the gap in resource allocation among TVs, resulting in reduced resource utilization fairness.} \textcolor{b1}{More specifically, the proposed JCRATOA outperforms the benchmarks in terms of task completion delay, task completion ratio, system throughput, and resource utilization fairness in relatively dense scenario, falling within the ranges of 7.6\% to 87\%, 6.6\% to 371\%, 6.25\% to 78.79\%, and 0.48\% to 87.23\% respectively.} This is because JCRATOA can dynamically offload tasks to the most suitable servers and stimulate resource sharing of FVs by using the two-sided matching method and the incentive mechanism. In contrast, methods like NRO, NFO, and NSO, which rely on the nearest server selection, \textcolor{b1}{struggle to maintain efficient task processing and fair resource utilization as the workload increases.} Moreover, the DRL-based methods of PTOCRA and DTOCRA exhibit evident disadvantages in terms of task processing performance and resource utilization fairness as the number of TVs increases. This is because they require extensive environmental interactions to learn effective policies since the coupled decision space expands rapidly with more TVs, thus resulting in higher computational complexity and longer delays.

\par From Fig. \ref{fig_tasknodes}(e), we observe that the average energy consumption of all approaches shows an overall upward trend as the number of TVs increases, as heavier workloads result in greater computing and uploading energy consumption. {\color{b1}However, the proposed JCRATOA exhibits higher energy consumption, which is a trade-off for achieving lower task completion delays, higher task completion ratios, higher system throughput and higher resource utilization fairness.} Comparatively, the lower energy consumption of the other approaches comes at the expense of task processing efficiency and delay. In summary, the results demonstrate that although the proposed JCRATOA incurs higher energy consumption as the number of TVs increases, it exhibits better scalability in efficiently handling delay-sensitive and computation-intensive tasks.

\subsubsection{Effect of Task Size}
\par {\color{b1}Figs. \ref{fig_tasksize}(a), \ref{fig_tasksize}(b), \ref{fig_tasksize}(c), \ref{fig_tasksize}(d), and \ref{fig_tasksize}(e) illustrate the impact of the task size on average task completion delay, average task completion ratio, system throughput, resource utilization fairness, and average energy consumption}, respectively, for different approaches. \textcolor{b1}{As shown in Figs. \ref{fig_tasksize}(a) to \ref{fig_tasksize}(d), as task size increases, the average task completion delay increases while the average task completion ratio, system throughput, and resource utilization fairness decline across all approaches. This is because larger tasks intensify both the processing and communication burdens, thus leading to longer delays and lower completion ratios under limited resources, while also reducing both throughput and fairness as computing resources become unevenly utilized among vehicles.} Moreover, the proposed JCRATOA achieves significantly superior performance in task processing efficiency and resource utilization fairness. \textcolor{b}{Specifically, in comparison to the other approaches, when the task size reaches 1000 KB, the JCRATOA achieves approximately 43\%, 29\%, 28\%, 22\%, 13\%, 35\%, 43\%, and 42\% performance gains in terms of the average task completion delay compared with ALO, NRO, NFO, NSO, KMMTO, BROLDRA, PTOCRA, and DTOCRA, respectively.} \textcolor{b1}{Additionally, the proposed JCRATOA significantly outperforms the other approaches in task completion ratio, system throughput, and resource utilization fairness, falling within the ranges of 12\% to 7150\%, 18.23\% to 98.99\%, and 27.79\% to 99.51\%, respectively.}

\par From Fig.~\ref{fig_tasksize}(e), we can observe that the average energy consumption of all approaches tends to increase as task size grows, since larger tasks require more computing resources, which consume more energy. Furthermore, although the DRL-based methods of PTOCRA and DTOCRA achieve lower energy consumption, \textcolor{b1}{the energy-saving advantage is obtained at the expense of inferior task processing performance and reduced resource utilization fairness.} In contrast, compared to the comparative approaches, \textcolor{b1}{the proposed JCRATOA exhibits higher energy consumption to maintain relatively lower task completion delay, higher task completion ratio, greater system throughput, and better resource utilization fairness.} In conclusion, the simulation results indicate that the proposed JCRATOA can adapt to heavy-loaded VFC scenarios, thus achieving superior performances of task completion delay, task completion ratio, system throughput, and resource utilization fairness within energy constraints, despite the trade-off of increased energy consumption.
                            

\subsubsection{Effect of FV Types}

\par To evaluate the effectiveness of the contract theory-based incentive mechanism of the proposed JCRATOA, we compare this incentive mechanism with the optimal contract mechanism and linear pricing mechanism~\cite{Kazmi2022ANovel} in Fig. 1 in Appendix O of the supplemental material. Specifically, the optimal contract mechanism under symmetric information (i.e., the MBS is aware of the types of FVs) serves as the upper bound of the MBS utility. Moreover, for the linear pricing mechanism under asymmetric information, the MBS lacks type information and only offers linear pricing options to the stimulate FVs to share the idle resources. The simulation results demonstrate that the contract-based incentive mechanism of the proposed JCRATOA can achieve balanced utility distribution between the MBS and FVs, thereby improving overall system efficiency.


{\color{b1}
\section{Discussion of the Performance Evaluation in a Hardware Environment}
\label{discussion}
\par To demonstrate the feasibility and effectiveness of the proposed approach in a real-world environment, we conduct experiments on an in-vehicle rugged computer Nuvo-9200VTC. The details are presented in Appendix P of the supplementary material. The results indicate that the proposed JCRATOA can operate efficiently on actual hardware and achieve satisfactory performance in terms of task completion delay, task completion ratio, system throughput, and resource utilization fairness under the energy constraints.
}

\section{CONCLUSION}
\label{CONCLUSION}

\par In this work, we have studied joint computing resource allocation and task offloading in VFC system. First, we have presented a hierarchical VFC architecture under asymmetric information, which integrates the computing capabilities of both RSUs and FVs. Moreover, we have formulated the DMOP to minimize the task completion delay of all TVs, while satisfying the energy constraints of TVs, RSUs, and FVs. To solve the DMOP, we have proposed the JCRATOA, which compromises the computing resource allocation and task offloading components. Specifically, we have presented a convex optimization-based method for RSU resource allocation and a contract theory-based incentive mechanism for FV resource allocation. Then, we have designed a two-sided matching method based on matching game to optimize task offloading decisions. \textcolor{b1}{Simulation results have demonstrated that the proposed JCRATOA achieves superior performance in terms of task completion delay, task completion ratio, system throughput, and resource utilization fairness while satisfying the energy constraints of different nodes.} Moreover, the proposed JCRATOA has better scalability in dense vehicular environments and demonstrates superior performance under heavy-loaded scenarios, achieving enhanced performance in terms of task completion delay and task completion ratio while adhering to energy constraints.

\par \textcolor{b}{The main limitation of the proposed JCRATOA is that it is applicable in urban areas where RSUs are readily accessible or easily deployed to serve as terrestrial MEC servers. However, JCRATOA may not be well-suited for remote, rural, mountainous, or hazardous areas where ground infrastructure deployment is challenging or impractical. Therefore, our future work will focus on joint optimization of computing resource allocation, task offloading, and UAV trajectory planning in uncrewed aerial vehicles (UAV)-assisted VFC systems by leveraging the flexibility, mobility, and line-of-sight communication of UAVs.}

\bibliographystyle{IEEEtran}
\bibliography{IEEEabrv,references}
\begin{IEEEbiography}[{\includegraphics[width=1in,height=1.23in,clip,keepaspectratio]{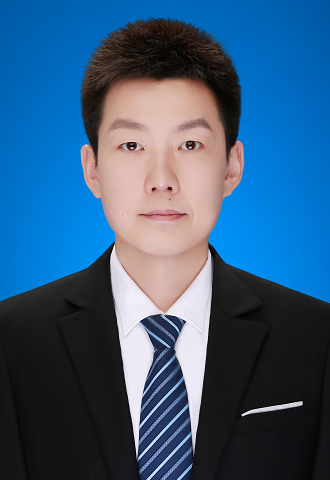}}]{Geng Sun} (Senior Member, IEEE) received the B.S. degree in communication engineering from Dalian Polytechnic University, and the Ph.D. degree in computer science and technology from Jilin University, in 2011 and 2018, respectively. He was a Visiting Researcher with the School of Electrical and Computer Engineering, Georgia Institute of Technology, USA. He is a Professor in the College of Computer Science and Technology at Jilin University. Currently, he is working as a visiting scholar at the College of Computing and Data Science, Nanyang Technological University, Singapore. He has published over 100 high-quality papers, including IEEE TMC, IEEE JSAC, IEEE/ACM ToN, IEEE TWC, IEEE TCOM, IEEE TAP, IEEE IoT-J, IEEE TIM, IEEE INFOCOM, IEEE GLOBECOM, and IEEE ICC. He serves as the Associate Editors of IEEE Communications Surveys \& Tutorials, IEEE Transactions on Communications, IEEE Transactions on Vehicular Technology, IEEE Transactions on Network Science and Engineering, IEEE Transactions on Network and Service Management and IEEE Networking Letters. He serves as the Lead Guest Editor of Special Issues for IEEE Transactions on Network Science and Engineering, IEEE Internet of Things Journal, IEEE Networking Letters. He also serves as the Guest Editor of Special Issues for IEEE Transactions on Services Computing, IEEE Communications Magazine, and IEEE Open Journal of the Communications Society. His research interests include Low-altitude Wireless Networks, UAV communications and Networking, Mobile Edge Computing (MEC), Intelligent Reflecting Surface (IRS), Generative AI and Agentic AI, and deep reinforcement learning.
\end{IEEEbiography}


 \begin{IEEEbiography}[{\includegraphics[width=1in,height=1.23in,clip,keepaspectratio]{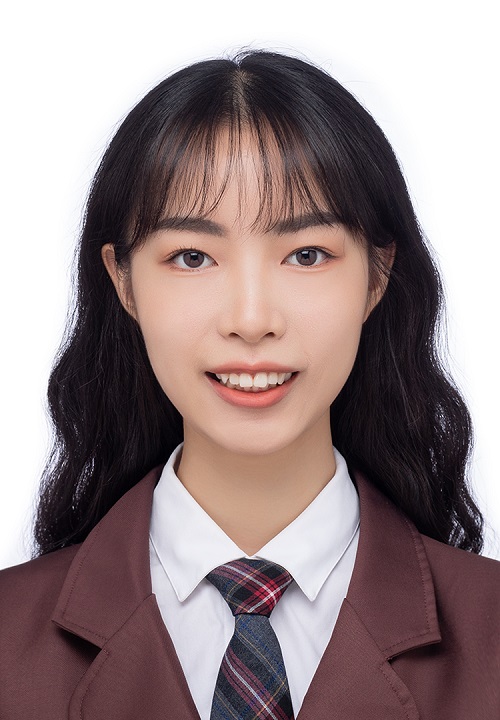}}]{Siyi Chen} received a BS degree in Computer Science and Technology from Harbin University of Science and Technology, Harbin, China, in 2023. She is currently working toward the MS degree in Software Engineering at Jilin University, Changchun, China. Her research interests include mobile edge computing and optimizations.
\end{IEEEbiography}

 
\begin{IEEEbiography}[{\includegraphics[width=1in,height=1.23in,clip,keepaspectratio]{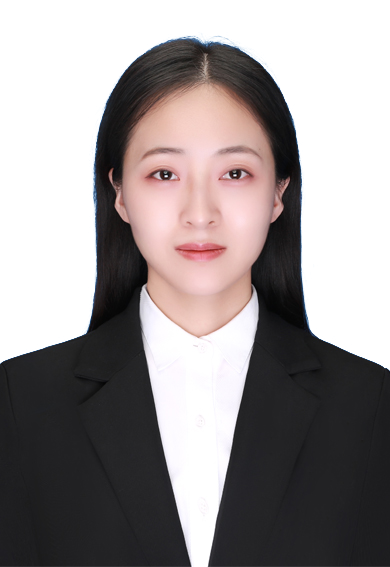}}]{Zemin Sun} received a BS degree in Software Engineering, an MS degree and a Ph.D degree in Computer Science and Technology from Jilin University, Changchun, China, in 2015, 2018, and 2022, respectively. Her research interests include mobile edge computing, UAV communications, and game theory. 
\end{IEEEbiography}

 
\begin{IEEEbiography}[{\includegraphics[width=1in,height=1.23in,clip,keepaspectratio]{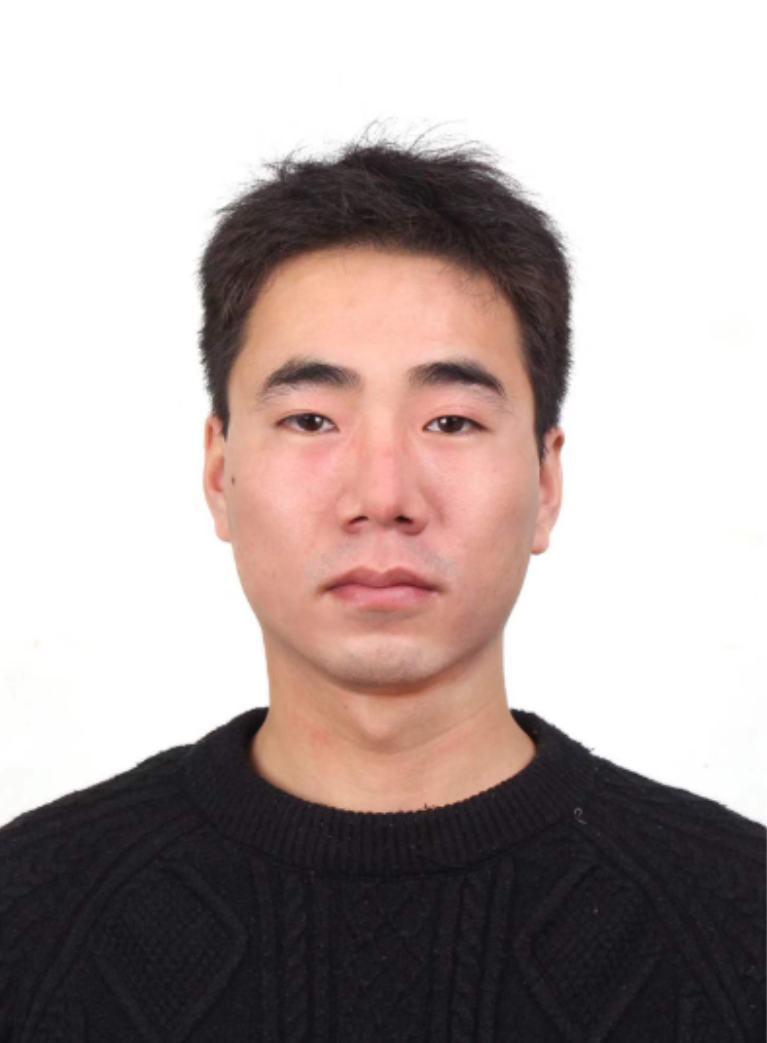}}]{Long He} received a BS degree in Computer Science and Technology from Chengdu University of Technology, Sichuan, China, in 2019. He is currently working toward the PhD degree in Computer Science and Technology at Jilin University, Changchun, China. His research interests include vehicular networks and edge computing.
\end{IEEEbiography}

 
\begin{IEEEbiography}[{\includegraphics[width=1in,height=1.23in,clip,keepaspectratio]{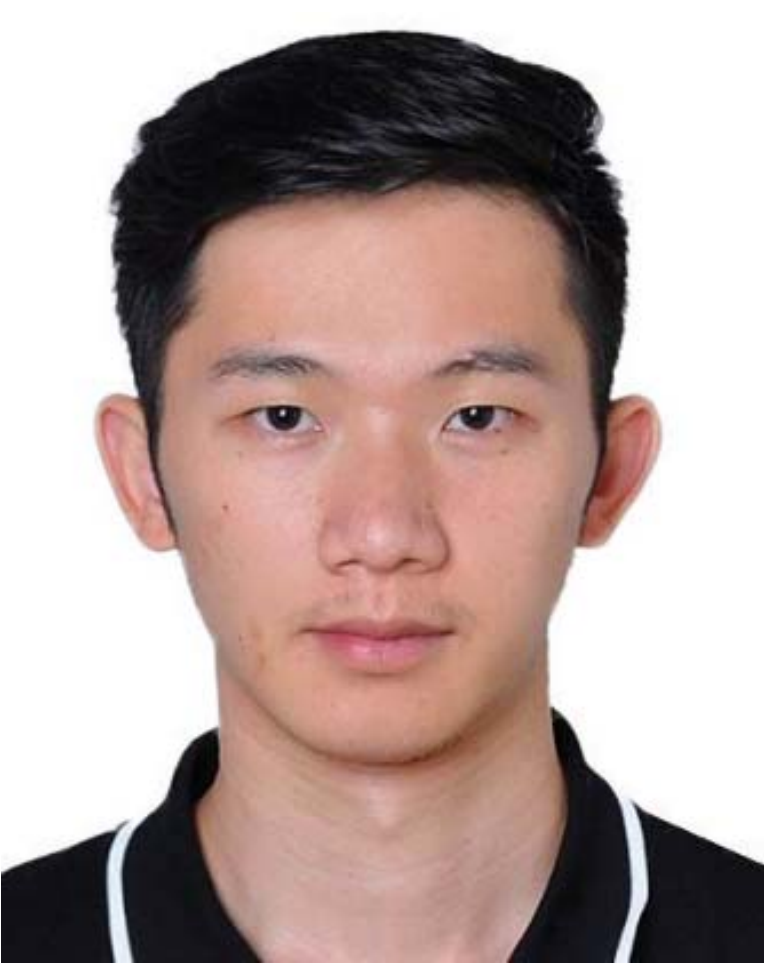}}]{Jiacheng Wang} received the Ph.D. degree from the School of Communication and Information Engineering, Chongqing University of Posts and Telecommunications, Chongqing, China. He is currently a Research Associate in computer science and
engineering with Nanyang Technological University,
Singapore. His research interests include wireless
sensing, semantic communications, and metaverse.
\end{IEEEbiography}


\begin{IEEEbiography}[{\includegraphics[width=1in,height=1.23in,clip,keepaspectratio]{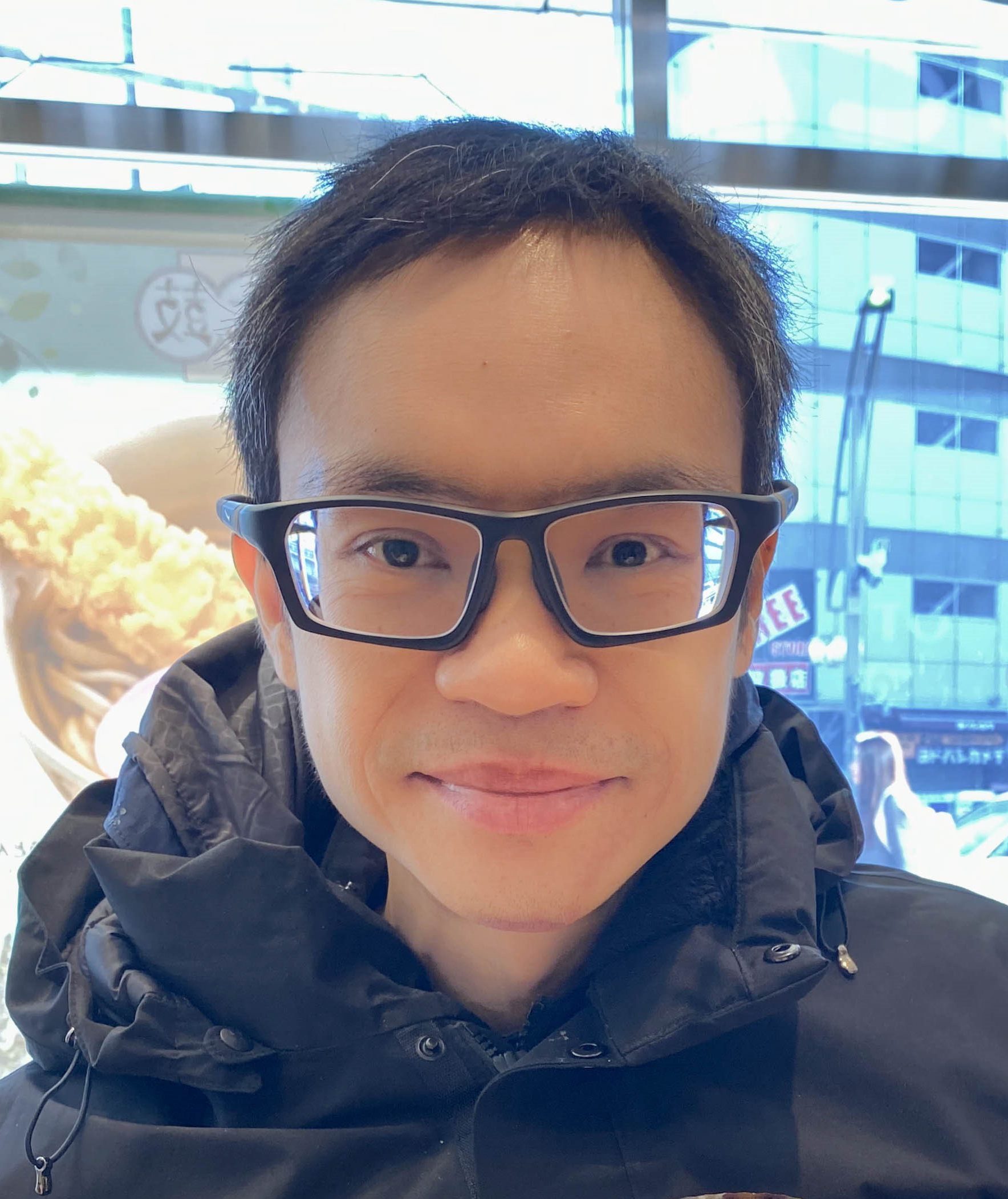}}]{Dusit Niyato}  (Fellow, IEEE) is a professor in the College of Computing and Data Science, at Nanyang Technological University, Singapore. He received B.Eng. from King Mongkuts Institute of Technology Ladkrabang (KMITL), Thailand and Ph.D. in Electrical and Computer Engineering from the University of Manitoba, Canada. His research interests are in the areas of mobile generative AI, edge intelligence, quantum computing and networking, and incentive mechanism design.
\end{IEEEbiography}


\begin{IEEEbiography}[{\includegraphics[width=1in,height=1.23in,clip,keepaspectratio]{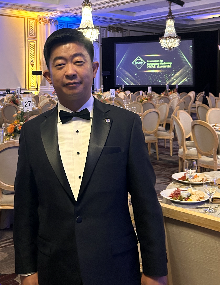}}]{Zhu Han} (S'01–M'04-SM'09-F'14) received the B.S. degree in electronic engineering from Tsinghua University, in 1997, and the M.S. and Ph.D. degrees in electrical and computer engineering from the University of Maryland, College Park, in 1999 and 2003, respectively. From 2000 to 2002, he was an R\&D Engineer of JDSU, Germantown, Maryland. From 2003 to 2006, he was a Research Associate at the University of Maryland. From 2006 to 2008, he was an assistant professor at Boise State University, Idaho. Currently, he is a John and Rebecca Moores Professor in the Electrical and Computer Engineering Department as well as in the Computer Science Department at the University of Houston, Texas. Dr. Han's main research targets on the novel game-theory related concepts critical to enabling efficient and distributive use of wireless networks with limited resources. His other research interests include wireless resource allocation and management, wireless communications and networking, quantum computing, data science, smart grid, carbon neutralization, security and privacy. Dr. Han received an NSF Career Award in 2010, the Fred W. Ellersick Prize of the IEEE Communication Society in 2011, the EURASIP Best Paper Award for the Journal on Advances in Signal Processing in 2015, IEEE Leonard G. Abraham Prize in the field of Communications Systems (best paper award in IEEE JSAC) in 2016, IEEE Vehicular Technology Society 2022 Best Land Transportation Paper Award, and several best paper awards in IEEE conferences. Dr. Han was an IEEE Communications Society Distinguished Lecturer from 2015 to 2018 and ACM Distinguished Speaker from 2022 to 2025, AAAS fellow since 2019, and ACM Fellow since 2024. Dr. Han is a 1\% highly cited researcher since 2017 according to Web of Science. Dr. Han is also the winner of the 2021 IEEE Kiyo Tomiyasu Award (an IEEE Field Award), for outstanding early to mid-career contributions to technologies holding the promise of innovative applications, with the following citation: ``for contributions to game theory and distributed management of autonomous communication networks."
\end{IEEEbiography}


\begin{IEEEbiography}[{\includegraphics[width=1in,height=1.23in,clip,keepaspectratio]{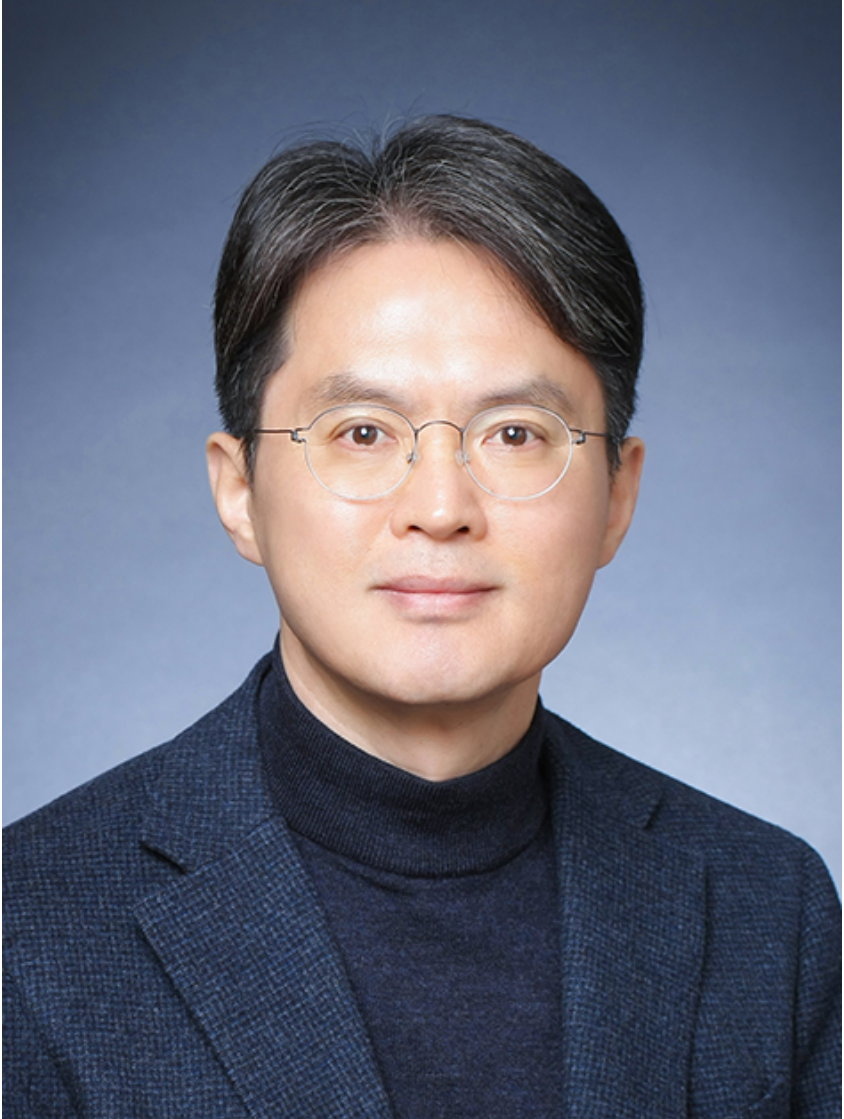}}]{Dong In Kim} (Fellow, IEEE) received the Ph.D. degree in electrical engineering from the University of Southern California, Los Angeles, CA, USA, in 1990. He was a Tenured Professor with the School of Engineering Science, Simon Fraser University, Burnaby, BC, Canada. He is currently a Distinguished Professor with the College of Information and Communication Engineering, Sungkyunkwan University, Suwon, South Korea. He is a Fellow of the Korean Academy of Science and Technology and a Member of the National Academy of Engineering of Korea. He was the first recipient of the NRF of Korea Engineering Research Center in Wireless Communications for RF Energy Harvesting from 2014 to 2021. He received several research awards, including the 2023 IEEE ComSoc Best Survey Paper Award and the 2022 IEEE Best Land Transportation Paper Award. He was selected the 2019 recipient of the IEEE ComSoc Joseph LoCicero Award for Exemplary Service to Publications. He was the General Chair of the IEEE ICC 2022, Seoul. Since 2001, he has been serving as an Editor, an Editor at Large, and an Area Editor of Wireless Communications I for IEEE Transactions on Communications. From 2002 to 2011, he served as an Editor and a Founding Area Editor of Cross-Layer Design and Optimization for IEEE Transactions on Wireless Communications. From 2008 to 2011, he served as the Co-Editor- in-Chief for the IEEE/KICS Journal of Communications and Networks. He served as the Founding Editorin-Chief for the IEEE Wireless Communications Letters from 2012 to 2015. He has been listed as a 2020/2022 Highly Cited Researcher by Clarivate Analytics.
\end{IEEEbiography}

\end{document}